\documentclass[11pt,a4paper,english]{article}

\usepackage{amsmath}
\usepackage{latexsym}
\usepackage[english]{babel}
\usepackage{amsfonts}
\usepackage[T1]{fontenc}
\usepackage[latin1]{inputenc}
\usepackage{amsthm}
\usepackage{amssymb}
\usepackage{version}
\usepackage{fancyhdr}
\usepackage{caption}
\usepackage{stmaryrd}
\usepackage{url,color}
\usepackage{hyperref}
\usepackage[numbers]{natbib}

\definecolor{darkgreen}{rgb}{0,0.5,0}

\usepackage{mathtools}
\usepackage{eurosym}
\usepackage{nicefrac}
\mathtoolsset{showonlyrefs}

\numberwithin{equation}{section} \makeatletter
\renewcommand{\subsection}{\@startsection
{subsection}{2}{0mm}{\baselineskip}{-0.25cm}
{\normalfont\normalsize\bf}} \makeatother

\newtheorem{theorem}{Theorem}[section]
\newtheorem{lemma}[theorem]{Lemma}

\newtheorem{proposition}[theorem]{Proposition}
\theoremstyle{definition}
\newtheorem{definition}[theorem]{Definition}
\theoremstyle{remark}
\newtheorem{remark}[theorem]{Remark}
\newtheorem{example}[theorem]{Example}
\newtheorem{assumption}[theorem]{Assumption}

\def \A {\mathcal A}

\def \F {\mathcal F}

\def \H {\mathcal H}

\def \S {\mathcal S}

\def \cT {\mathcal T}

\def \P {\mathbf P}
\def \Q {\mathbf Q}
\def \I {{\mathbf 1}}

\def \R {\mathbb R}
\def \bF {\mathbb F}

\def \bE {\mathbb E}
\def \bN {\mathbb N}

\def \Nu {\boldsymbol{\nu}}
\def \bpi{\boldsymbol{\pi}}

\newcommand{\ud}{\mathrm d}

\newcommand{\ds}{\displaystyle}
\newcommand{\esp}[2][\mathbb E] {#1\left(#2\right)}
\newcommand{\qesp}[2][\mathbb E^{\mathbf Q}] {#1\left(#2\right)}
\newcommand{\numin}{0}
\newcommand{\numax}{\nu^{\text{max}}}
\newcommand{\cdown}{c^{\text{down}}}
\newcommand{\cupp}{c^{\text{up}}}

\newcommand{\wstep}{m_1}
\newcommand{\pistep}{m_2}
\newcommand{\fconst}{c_f}
\newcommand{\fexp}{\varsigma}
\newcommand{\Vdet}{V^{\prime,\text{det}}}
\newcommand{\pipost}{\pi^{\text{post}}}

\begin{document}

\title{Optimal Liquidation under Partial Information with Price Impact}




\author{ Katia Colaneri \footnote{Department of Economics, University of Rome Tor Vergata, Via Columbia 2, 00133 Roma, Italy  \texttt{katia.colaneri@uniroma2.it}}
\and Zehra Eksi \footnote{Institute for Statistics and Mathematics, Vienna University of Economics and Business (WU), Welthandelsplatz 1, 1020 Vienna, Austria, \texttt{zehra.eksi@wu.ac.at}}
\and R\"{u}diger Frey \footnote{Institute for Statistics and Mathematics, Vienna University of Economics and Business (WU), \texttt{ruediger.frey@wu.ac.at}, corresponding author}
\and Michaela Sz\"{o}lgyenyi \footnote{Department of Statistics, University of Klagenfurt, Universit\"atsstra\ss{}e 65-67, 9020 Klagenfurt, Austria, \texttt{michaela.szoelgyenyi@aau.at}}}

\date{}

\maketitle

\begin{abstract}
We study the optimal liquidation problem in a market model where the bid price follows a geometric pure jump process whose  local characteristics are driven by an unobservable finite-state Markov chain and by the liquidation rate. This model is consistent with  stylized facts of high frequency data such as the discrete nature of tick data and  the clustering in the order flow.  We include both temporary and permanent effects into our analysis.  We  use stochastic filtering to reduce the optimal liquidation problem to an equivalent optimization problem  under complete information. This leads to a stochastic control problem for piecewise deterministic Markov  processes (PDMPs). We carry out a detailed mathematical analysis of this problem. In particular, we derive the optimality equation for the value function, we characterize the value function as continuous viscosity solution of the associated dynamic programming equation, and we prove a novel comparison result.
The paper concludes with numerical results illustrating the impact  of partial information and price impact on the value function and on the optimal liquidation rate.


\end{abstract}



\vspace{0.2cm}

\section{Introduction}

In financial markets, traders  frequently face  the task of selling a large amount of a given asset over a  short time period. This has led to a lot of research on optimal portfolio execution, largely in the context of  market impact models. In these models one directly specifies the  impact of a given trading strategy on the bid price of the asset and the fundamental price (i.e. the price if the trader is inactive) is usually modelled as a diffusion process. However, portfolio liquidation strategies are executed at relatively high trading frequency. Hence a sound market impact model should be consistent with   key stylized facts of high frequency data as discussed  for instance by  \citet{bib:cartea-jaimungal-penalva-15} or \citet{bib:cont-11}.  First,  on   fine time scales  the bid price of an asset is  best described by a pure jump   process, since in reality   prices move on a discrete grid defined by the tick size.
Second, the order flow is clustered in time: there are random periods with  a lot of  buy orders  or  with a lot of  sell orders, interspersed by quieter times with less trading activity. \citet{bib:cont-11} attributes this to  the fact that  many observed orders are components of a larger parent order that is executed in small blocks.   A further reason for the clustering in the inter-event times  are  random  fluctuations in the arrival rate of new information, see, e.g.~\citet{bib:andersen-96}.
Third, the distribution of returns over short time intervals is strongly non-Gaussian but has heavy tails and a large mass around zero; to a certain extent this is a consequence of the first two stylized facts.
Finally,  there is permanent price impact,   that is  the implementation of a liquidation strategy pushes prices downwards.

To capture these  stylized facts we  model the bid price as a  marked point process with Markov switching whose local  characteristics (intensity and jump size distribution)  depend on  the trader's current liquidation rate $\nu_t$  and on the value $Y_t$ of a finite state  Markov chain $Y$. The fact that  the local characteristics depend on $\nu_t$  is used to model permanent price impact. Markov switching allows us to reproduce the  observed clustering in the order flow.   Our framework encompasses models with  a high intensity of downward jumps in one state of $Y$ and  a high intensity of upward jumps in another state of $Y$ and models where inter-event times are given by a mixture of exponential distributions. We view the process $Y$ as an abstract modelling device that  generates clustering and assume therefore that $Y$ is  unobservable by the trader. This is consistent with the fact that economic sources for clustering such as the trading activity of other investors are  not directly observable.
Markov modulated marked point processes  with partial information (without price impact) were considered previously in the statistical modelling of high frequency data, see for instance   \citet{bib:cvitanic-rozovski-zalyapin-05} or \citet{bib:cartea-jaimungal-13}; however, we are the first to study optimal liquidation  in such a setting.

The first step in the analysis of a  control problem with  partial information is to formulate an equivalent  problem under full information via stochastic filteringand hence  to derive the stochastic filtering equations for our setup. These equations describe the dynamics of the  conditional distribution of $Y_t$ given the bid price history up to time $t$. Note that this provides a further rationale for modelling the bid price as a marked  point process.  In fact, the strong non-normality  of short-period  returns implies that it is  problematic to use high frequency data as input in the numerical solution of the  classical filtering equations for models with diffusive observation process, as   the resulting filters may become   unstable.  To overcome this issue we prefer to work in a point process framework.
We use the reference probability approach to give a rigorous construction of our model and to derive the filtering equations.
We end up with a control problem whose state process, denoted by $X$, consists of the stock price, the inventory level, and the filter process.
We provide a detailed mathematical analysis of this  problem. The form of the asset return dynamics implies  that $X$ is a piecewise deterministic Markov process (PDMP) so that we rely on control theory for PDMPs;  a general  introduction to this theory is given in  \citet{davis1993markov} or in \citet{bauerle-rieder-book}.
We establish the dynamic programming equation for the value function and we derive  conditions on the data of the problem that guarantee the continuity of the value function. This requires a careful analysis of the behaviour of the value function close to the boundary of the state space.
As a further step we  characterize the value function as the unique continuous  viscosity solution of the Hamilton-Jacobi-Bellman (HJB) partial integro-differential equation   and we give an example showing that in general the HJB equation does not admit a classical solution.  Moreover,  we prove a novel comparison theorem for the HJB equation which is valid in more general PDMP setups. A comparison principle is necessary to  ensure  the convergence of  numerical schemes to the value function, see \citet{barles1991}.

The paper closes with a section on applications. We discuss   properties  of the optimal liquidation rate and of the expected liquidation profit and we use a finite difference approximation of the HJB equation to  analyze the influence of the temporary and permanent price impact parameters on the form of the optimal liquidation rate. Among others,  we find that for certain parameter constellations the optimal strategy displays a surprising  gambling behaviour of the trader that cannot be guessed upfront and we give an economic interpretation that is based on the form of the HJB equation. Moreover, we study the  additional liquidation profit from the use of a  filtering model, and  we report results from a small calibration study that provides further support for our model.

We continue with a brief discussion  of the existing literature, concentrating on market impact models.
The first contribution is \citet{bertsimas1998optimal} who analyze  the optimal portfolio execution problem for a risk-neutral agent in a  model with linear and purely permanent price impact. This model has been generalized by \citet{almgren2001optimal}  who consider also  risk aversion and  temporary price impact. Since then, market impact models have been extensively studied. Important contributions include 
\citet{he2005dynamic}, \citet{schied2009risk},  \citet{schied16}, \citet{guo2015optimal}, \citet{bib:casgrain-jaimungal-19}. All these models work in a (discretized) diffusion framework. From a methodological point of view our analysis is also related to the literature on expected utility maximization  for pure jump process
such as \citet{bib:bauerle-rieder-09}.
Important contributions to the control theory of PDMPs include   \citet{davis1993markov, almudevar2001dynamic,costa2013continuous}. Viscosity solutions for PDMP control problems were previously considered in \citet{bib:davis-farid-99}.

The outline of the paper is the following. In Section \ref{sec:framework}, we introduce our model, the main assumptions and the optimization problem. In Section \ref{sec:partial info}, we  derive the filtering equations for our model. Section \ref{section:mdm} contains the mathematical analysis of the  optimization problem via PDMP techniques. In Section \ref{sec:viscosity} we provide a viscosity solution characterization of the value function.
Finally, in Section \ref{sec:numerics}, we present the results of our numerical experiments.
The appendix  contains additional  proofs. 

\section{The Model}\label{sec:framework}

\subsection{The optimal liquidation problem}

Throughout we work on the filtered probability space $(\Omega,\mathcal{F},\mathbb{F},\mathbf{P})$, where the filtration $\mathbb{F}=\{\mathcal{F}_t\}_{t\geq 0}$ satisfies the usual conditions. Here $\mathbb{F}$ is the global filtration, i.e. all considered processes are $\mathbb{F}$-adapted, and $\mathbf{P}$ is the historical probability measure.
We consider a trader who has to liquidate  $w_0>0$ units of a given security (referred to as the \emph{stock} in the sequel) over the period $[0, T]$ for a given time horizon $T$. We denote the bid price process by $S = (S_t)_{0 \leq t \le T}$ and  $\mathbb{F}^S$ is the right-continuous and complete filtration generated by $S$.

We assume that the trader  sells the shares at a {nonnegative} $\mathbb{F}^S$-adapted rate $ \Nu = (\nu_t)_{0 \le t \le T}$.
Hence  her inventory, i.e.~the amount of shares she holds at time $t \in [0,T]$,  is given by the absolutely continuous process
\begin{equation}\label{eq:inventory}
W_t = w_0 -\int_0^t \nu_{u}\ud u, \quad 0 \le t \le T.
\end{equation}
The inventory dynamics \eqref{eq:inventory} is a stylized model of  real trading where orders  are placed at discrete points in time.   Our interpretation follows the literature on price impact models  such as   \citet{almgren2001optimal} or \citet{bib:cartea-jaimungal-penalva-15}, Section~6.2. We split the time interval $[0,T]$  into small subintervals of fixed length $\delta$, leading to a partition $0=t_0<t_1< \dots <t_n=T$.
At each time $t_j$, the investor decides on the amount of inventory  she wants to liquidate over the period $[t_j, t_{j+1})$. This quantity is described in terms of the $\F_{t_j}^S$-measurable nonnegative  trading rate  $\nu_{t_j} = (W_{t_{j}}-W_{t_{j+1}}) / \delta$, $j =0, \dots, n-1 $.  We assume that the revenue generated by this share sale is given by
\begin{equation}\label{eq:revenue-from-liquidation}
(\nu_{t_j} \delta) S_{t_j}(1-f(\nu_{t_j})), \quad j =0, \dots, n-1.
\end{equation}
Here $S_{t_j}(1-f(\nu_{t_j}))$ is the execution price per share, and the nonnegative, continuous  and increasing function $f$ models  \emph{temporary price impact}.  A simple interpretation of \eqref{eq:revenue-from-liquidation}  is that $\nu_{t_j} \delta$ shares are sold in a market order, and the quantity $f(\nu_{t_j})$ describes in a stylized way the impact on the execution price as the order ``walks the order book''. More abstractly, one may  view \eqref{eq:revenue-from-liquidation} as  (expected) revenue of some ultra-high frequency trading algorithm for the  liquidation of  the  child order $\nu_{t_j} \delta $ over $(t_j, t_{j+1}]$, see for instance \citet{lehalle2018optimal}. Note that the price impact described by  $f$ is purely temporary as it pertains only to the execution price of the current trade. Permanent price impact (the impact of trading on the dynamics of $S$) is discussed in the next section.

Consider now a discrete list of share sales  $\{\nu_{t_0},\dots,\nu_{t_{n-1}}\}$ and define the associated continuous-time  liquidation strategy  $\Nu$ by
\begin{equation}\label{eq:piecewise-constant}
\nu_t = \sum_{j=0}^{n-1}  \nu_{t_j} 1_{(t_j,t_{j+1}]}(t)\,, \quad 0\le t \le T\,.
 \end{equation}
Then for small $\delta$ the inventory process generated by the discrete trades is  approximately equal to  \eqref{eq:inventory},  and the cumulative  revenue process of the discrete trades  is approximately equal to the absolutely continuous cash flow  stream
\begin{equation} \label{eq:cash-flow}
 \int_0^t \nu_s S_s (1-f(\nu_s)) \ud s\,, \quad 0 \le t \le T.
\end{equation}
In this paper we work with the absolutely continuous  inventory dynamics \eqref{eq:inventory} and with the  revenue stream~\eqref{eq:cash-flow}.  This  facilitates comparison with the literature and permits us to  use  tools from stochastic calculus and continuous-time stochastic control.

Now we describe the ingredients of the liquidation problem in detail.  First, we  consider the temporary  price impact.  Empirical evidence suggests that $f$ can be modelled as a power function,  $f(\nu)=\fconst \nu^\fexp$, with $0 < \fexp <  1$, see for instance \citet{bib:cartea-jaimungal-penalva-15}, Section~6.7 or \citet{almgren}. In that case $\nu f(\nu)$, the  cost of trading at the rate $\nu$, is increasing, strictly convex and exhibits superlinear growth, that is  $\lim_{\nu \to \infty} f(\nu) = \infty$. In  the theoretical part of our analysis we do not specify explicitly the functional form of $f$,  but we assume throughout that  $\nu f(\nu)$ is increasing and  strictly convex  with superlinear growth.
Second, in order to  account for the case where not all shares have been sold strictly prior to time $T$, we model the value of the  remaining share position by  $h(W_T)S_T$. Here $h$ is an increasing, continuous and  concave function  with $h(w) \le w$ and $h(0)=0$. We also view the difference  $(W_T - h(W_T))S_T$ as a \emph{penalization} of a nonzero terminal inventory position.
Third, we confine  the trader to $\mathbb{F}^S$ adapted strategies. Moreover,  observe that  convexity and superlinear growth of   $\nu f(\nu)$ imply that the mapping $\nu \mapsto \nu  S_{t} (1-f(\nu))$
(the instantaneous  revenue generated by a share sale at rate $\nu$) has a unique maximum at some  $\numax >0$.
Hence  it is  never optimal for the agent to liquidate shares at a rate larger than $\numax$. We may therefore assume without loss of generality that  the liquidation rate is bounded by $\numax$, and we call a liquidation strategy $\Nu =(\nu_t)_{0 \le t \le T} $ \emph{admissible} if $\Nu$  is $\mathbb{F}^S$ adapted and  if $\nu_t \in [0, \numax]$ for all $0 \le t \le T$.\footnote{Imposing an  upper bound on the liquidation rate facilitates the mathematical analysis, since  existing results on  the control of piecewise deterministic Markov processes rely on the assumption of  a compact control space.}
The exact form of $f$  for  large $\nu$ large  is hard to estimate empirically, since one needs to extrapolate beyond the typical order size. Consequently it is difficult to estimate the value of $\numax$ precisely.  In Section~\ref{subsec:bounds} we therefore show that the value of the optimal liquidation problem is fairly insensitive to the exact  value of the this parameter.

Finally we describe the objective of the trader.  Fix some admissible liquidation strategy $\Nu$ and denote by $\rho \ge 0$ the (subjective) discount rate of the trader.   The expected discounted  value of the revenue generated by  $\Nu$  is  equal to
\begin{equation}\label{eq:optimization_full}
J(\Nu) = \mathbb{E}\left (\int_0^\tau e^{-\rho u} \nu_u S_u^{\Nu} (1 -f(\nu_{u})) \ud u + e^{-\rho \tau } S_\tau^{\Nu} h( W_\tau) \right )\,.
\end{equation}
Here  $S^{\Nu}$ denotes the bid price given that the trader follows the strategy $\Nu$ (see Section~\ref{subsec:bid-price}), and the $\bF^S$ stopping time $\tau$ is given by
\begin{gather}\label{tau1}
\tau:=\inf\{t\geq 0 : \ W_t\leq 0\}\wedge T\,.
\end{gather}
The goal of the trader is to maximize \eqref{eq:optimization_full} over all admissible strategies; the corresponding  optimal value is denoted by $J^*$.

Note that the form of the objective function in \eqref{eq:optimization_full} implies that the trader is risk neutral. Risk neutrality seems a reasonable assumption in our setup  since the time period $[0,T]$ is fairly short and since the trader is confined to pure selling strategies so that the  risk she may take is limited. A straightforward way to incorporate risk aversion  into our model  is to include a penalty of the form $ - \int_0^\tau   e^{-\rho u} S_u W_u \ud u$ into the reward function. Such term penalizes slow execution and hence strategies with high price risk.
Similar ideas are discussed for instance in \citet{bib:cartea-jaimungal-penalva-15}, Section 6.5.

%
%

\subsection{Dynamics of the bid price.}\label{subsec:bid-price}

To capture the discrete nature of high-frequency price trajectories we  model the bid price as a Markov-modulated geometric finite activity pure jump process. Let $Y = (Y_{t})_{0 \leq t \le T}$
be a continuous-time finite-state Markov chain  on $(\Omega,\mathcal{F},\mathbb{F},\mathbf{P})$ with state space $\mathcal{E} = \{e_1,e_2,...,e_K\}$ ($e_k$ is $k$-th unit vector in  $\mathbb{R}^K$),  generator matrix $ Q=(q^{ij})_{i,j=1,\dots,K}$ and  initial distribution $\pi_0=(\pi_0^1,\cdots,\pi_0^K)$.
The bid price has the dynamics
 \begin{equation}\label{eq:bid_price}
 \ud S_t = S_{t^-} \ud R_t, \quad S_0 = s\in(0, \infty)\,,
 \end{equation}
where the \emph{return process}   $R=(R_t)_{0 \leq t \le T}$ is a finite activity pure jump process. Moreover, it holds that $\Delta R_t:=R_t- R_{t^-} > - 1$ for all $t$, so that $S$ is strictly positive.  Note that $\mathbb{F}^S$ is equal to $\mathbb{F}^R$, the filtration generated by the return process $R$; in the sequel we will use both filtrations interchangeably.  Denote by $\mu^R$ the random measure associated with $R$, defined by
\begin{equation} \label{eq:def-muR}
\mu^R(\ud t, \ud z):= \!\!\sum_{u \ge 0,  \Delta R_u\neq 0}\delta_{\{u, \Delta R_u\}}(\ud t, \ud z),\,
\end{equation}
and by $\eta^{\P}$ the $(\bF,\P)$-dual predictable projection (or compensating random measure) of $\mu^R$.  We assume that  $\eta^{\P}$  is absolutely continuous and of the form $\eta^\P(t,Y_{t^-},\nu_{t^-};\ud z) \ud t$, for a finite measure $\eta^\P(t,e,\nu;\ud z)$ on $\R$ and that  the processes $R$ and $Y$ have no common jumps. Hence $R$ and $Y$ are orthogonal, i.e.  $[R,Y]_t \equiv 0$  for all $t \in [0, T]$, $\P$-a.s.

The measure $\eta^\P(t,e,\nu;\ud z)$ is a crucial quantity as it determines the law of the bid price with respect to filtration $\bF$ under the probability $\P$.  The fact that $\eta^{\P} $ depends on the current liquidation rate serves to model permanent price impact; the dependence of $\eta^{\P} $ on $Y_{t-}$ can be used  to reproduce the clustering in inter-event durations observed in high frequency data and to  model the  feedback effect from the trading activity of the rest of the market.
Finally,  time-dependence of  $\eta^{\P} $  can be used to model the strong intra-day seasonality patterns observed for high frequency data. These aspects are explained in more detail in Example \ref{exm2} below.

Next we provide further motivation for our setup. The discrete nature of high frequency price trajectories is illustrated in Figure~\ref{fig:high-frequency-prices} where we plot of Google share price sampled at a one-second frequency, together with a QQ~plot of the corresponding returns. The latter plot clearly shows that the returns are strongly non-Gaussian.
\begin{figure}[ht]
\begin{center}
\includegraphics[width =6.8cm, height= 5.6cm]{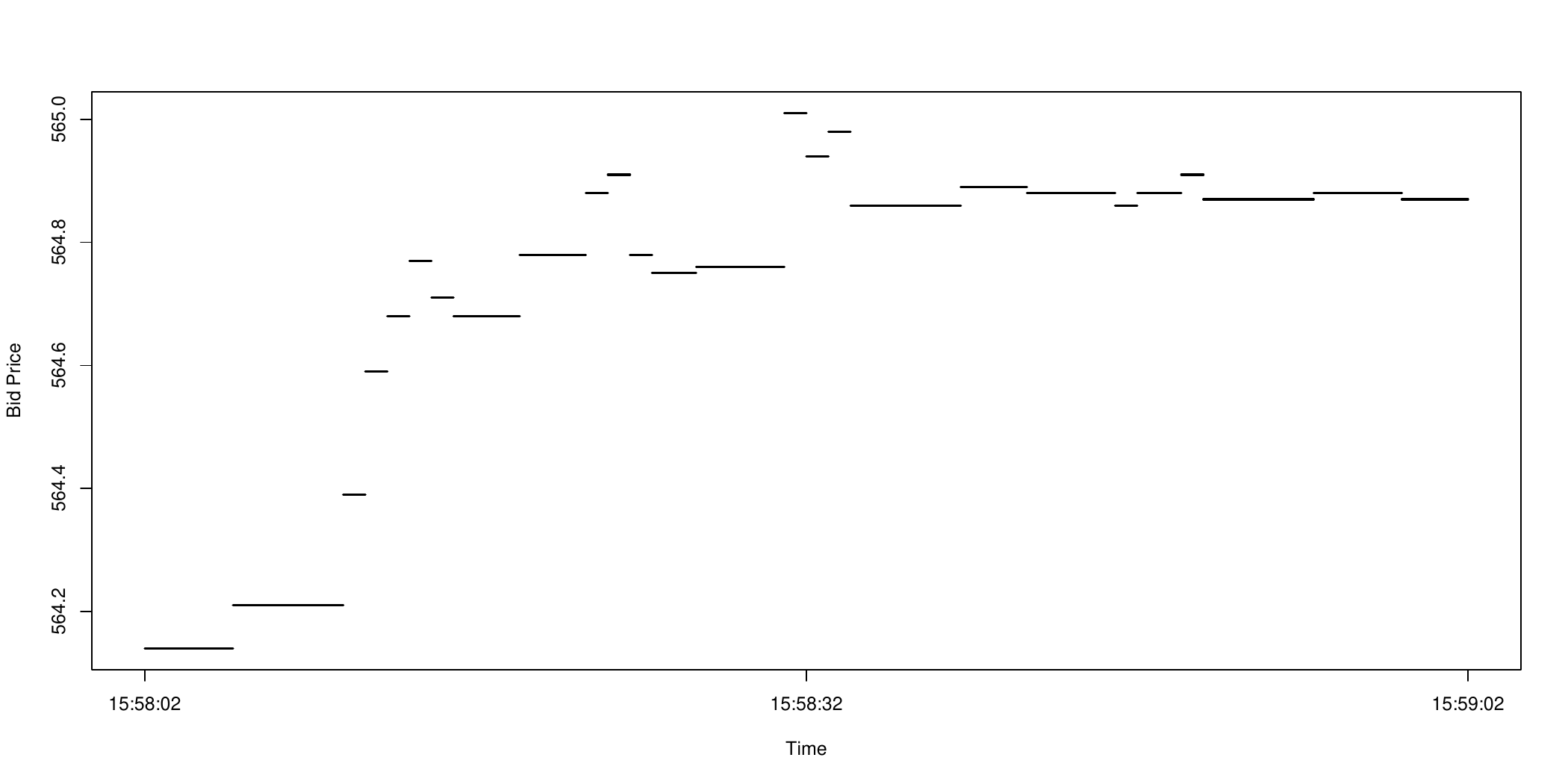} \quad
\includegraphics[width =6.8cm, height= 5.4cm]{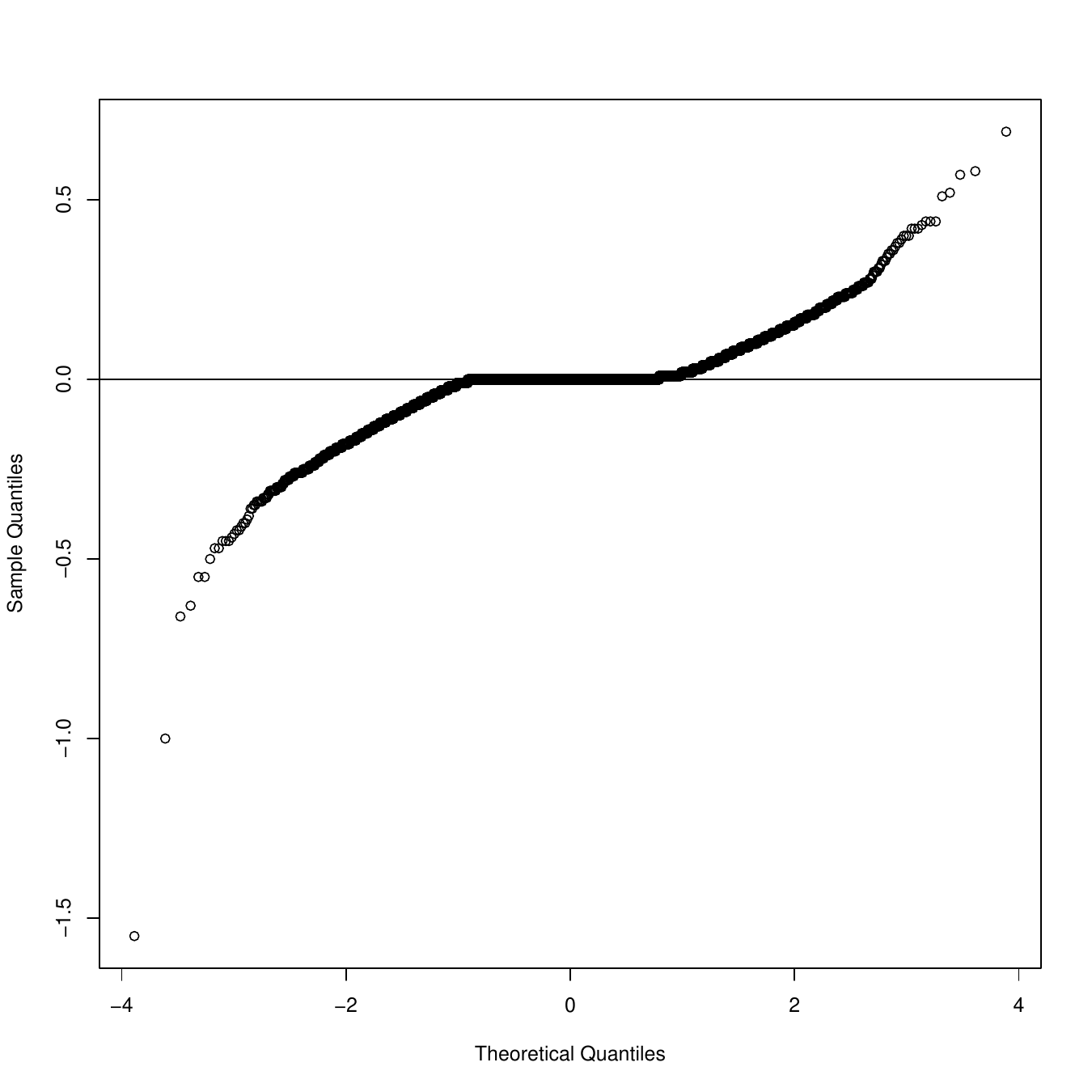}
\caption{Properties of high frequency returns. Left plot: Google share price on 2012-06-21 sampled between 15:58:02 and 15:59:02 at the frequency of one second; right plot: QQ~plot of the corresponding returns.   \label{fig:high-frequency-prices}}
\end{center}
\end{figure}
The  discretisation parameter  $\delta$ in \eqref{eq:revenue-from-liquidation} is typically of the order of a few minutes and hence larger than the time scale used in Figure~\ref{fig:high-frequency-prices},  so that on the $\delta$-timescale  a diffusion approximation for $S$ might  make sense.
In our setup with \emph{partial information} a point process  framework is nonetheless more appropriate, for the following reason.  As  the  process $Y$ is not directly observable, the optimal liquidation rate $\nu_t$ depends on the trader's  \emph{estimate} of the state of the market as  summarized by the conditional state probabilities $\pi_t^i=P(Y_t=e_i|\F^S_t)$, $i=1, \dots, K$. For instance, in the context of Example \ref{exm2} below it is intuitively clear  that if the conditional probability to be in the good state is high ($\pi^1_t$ close to one), the trader might want to wait in anticipation that the price will rise. In mathematical terms this means that  one has to add the filter process $(\pi^1_t, \dots, \pi^K_t )_{0\leq t \leq T} $ to the  state variables of the problem. The latter solves a $K$ dimensional SDE (called Kushner-Stratonovich equation) that is driven by the return observations. Here the following issue arises. In the numerical analysis of the optimal  liquidation problem  one needs to  solve this SDE  on a very fine time scale to make good use of the available information.\footnote{For instance we used one second returns for our calibration study in Section~\ref{subsec:calibration}.}  This leads to  numerical difficulties  if one works with  a diffusion model for $R$, essentially because   high frequency returns are strongly non-Gaussian.   For this reason we prefer to model the return  as  a marked point process.   Consistency between the model used for the computation of the filter process $\pi$ and the model used in the optimal liquidation problem itself thus implies that the latter problem should be analyzed  in a point process framework.

Further empirical support for our setup comes from a small calibration study with simulated  and real data presented in Section~\ref{subsec:calibration}. There we use the fact that for $\Nu   \equiv 0$ the model is a hidden Markov model with point process observation and we apply the expectation maximization (EM) methodology for Markov modulated point processes   to estimate the generator matrix of $Y$ and the parameters of the compensator $\eta^\P$.

It is interesting to consider the  \emph{semimartingale decomposition} of the bid price with respect to the full information filtration $\mathbb{F}$.  Denote for all $(t,e, \nu)\in[0,T]\times \mathcal E \times [0, \numax]$ the mean of ${\eta}^\P$ by
\begin{equation}\label{eq:drift}
\overline{\eta}^\P(t,e, \nu): = \int_{\R} z \, \eta^\P(t,e, \nu; \ud z)\,;
\end{equation}
$\overline{\eta}^\P(t,e, \nu)$ exists under Assumption~\ref{ass1} below. Fix some liquidation strategy $\Nu$. Then the martingale part $M^R$ of the return process is given by
$M^R_t = R_t - \int_0^t \overline{\eta}^\P (s,Y_{s-}, \nu_{s-})\ud s $, for all $t\in[0,T]$,  and  the $\mathbb{F}$-semimartingale decomposition of $S$ equals
\begin{equation}\label{eq:semimartingalefull}
S_t = S_0 + \int_0^t S_{s-} \ud M^R_s + \int_0^t  S_{s-} \overline{\eta}^\P(s,Y_{s-}, \nu_{s-})\ud s\,, \quad t\in[0,T]\,.
\end{equation}
In the sequel we assume that for all $(t,e) \in [0,T] \times \mathcal{E}$, the mapping $\nu \mapsto  \overline{\eta}^\P(t,e, \nu)$  is decreasing on $[0, \infty)$, that is the drift in the semimartingale decomposition \eqref{eq:semimartingalefull} of $S$ is decreasing in the liquidation rate. This is similar to the modeling of permanent price impact in \citet{almgren2001optimal} where liquidation adds a negative drift to bid price.

Finally we introduce some technical assumptions on the compensator $\eta^P$.
\begin{assumption} \label{ass1}
There  is a deterministic finite measure $\eta^{\Q}$ on $\R$ whose support, denoted by $\text{supp}(\eta)$, is  a compact subset of $(-1, \infty)$, such that for all $(t,e,\nu) \in[0,T] \times \mathcal{E}\times [0, \infty)$ the measure $\eta^{\P} (t,e, \nu; \ud z)$ is equivalent to   $\eta^{\Q}(\ud z) $.
     Furthermore,  for every $\numax < \infty$ there is some constant $M > 0$ such that
\begin{equation} \label{eq:cond-on-density}
    M^{-1} < \frac{\ud \eta^\P (t,e, \nu) }{\ud \eta^\Q}(z) < M  \text{ for all } (t, e, \nu)  \in[0,T] \times \mathcal{E} \times [0, \numax]\,.
\end{equation}
\end{assumption}
The assumption implies that for every fixed $\numax $ there is a  $ \lambda^{\text{max}}  < \infty$ such that
\begin{equation}\label{eq:bound_etaP}
\sup\{\eta^\P(t,e,\nu;\R) \colon (t,e,\nu) \in[0,T]\times \mathcal{E}\times[0 ,  \numax]\} \le  \lambda^{\text{max}} \,;
\end{equation}
in particular the counting process associated to the jumps of $S$ is $\P$-nonexplosive.
Moreover, it provides a sufficient condition for the existence of a \emph{reference probability measure}, i.e. a probability measure $\Q$ equivalent to $\P$ on $(\Omega, \F_T)$, such that under $\Q$, $\mu^R$ is a Poisson random measure with intensity  measure $\eta^{\Q}(\ud z)$,  independent of $Y$ and $\nu$. This is needed in the analysis of the filtering problem of the trader in Section~\ref{sec:partial info}. Note that the equivalence of $\eta^\P $ and $\eta^\Q$ implies that for all $(t,e,\nu) \in[0,T] \times \mathcal{E}\times [0, \infty)$ the support of $\eta^{\P}$ is equal to $\text{supp}(\eta)$. The assumption that $\text{supp}(\eta)$ is compact is not restrictive, since in reality the bid price moves only by a few ticks at a time.

The following examples serve to illustrate our framework; they  will be taken up in  Section~\ref{sec:numerics}.

\begin{example}\label{exm1}
Consider the case where the return process $R$ follows a bivariate  point process, i.e. there are two possible jump sizes, $\Delta R\in\{- \theta, \theta \}$ for some $\theta >0$. In this example we assume that the dynamics of $S$ is independent of  $Y$ and $t$. Moreover, the intensity $\lambda^+$ of an upward jump is constant and  equal to $\cupp >0$, and the intensity $\lambda^-$   of a downward jump depends on the rate of trading and is given by
$
\lambda^-(\nu) = c^{\text{down}}(1+ a \nu)
$
for constants  $c^{\text{down}}, a > 0$.
Note that, with this choice of $\lambda^-$, the intensity of a downward jump in $S$  is {linearly} increasing in the liquidation rate $\nu$. The  function $\overline{\eta}^\P$ from~\eqref{eq:drift} given by $\overline{\eta}^\P(\nu) = \theta (\cupp - \cdown(1+a \nu))$  is clearly independent of $t$ and $e$ and linearly decreasing in $\nu$.  Linear models for the permanent price  impact  are frequently considered in the literature as they have theoretical and empirical advantages; see for instance \citet{almgren} or \citet{gatheral2013dynamical}.
\end{example}

\begin{example}\label{exm2} Now we generalize Example~\ref{exm1} and  allow  $\eta^\P$ to depend  on the state process $Y$.  We consider a two-state Markov chain  $Y$ with the state space $\mathcal{E}=\{e_1,e_2\}$ and we assume that $e_1$ is a `good' state and $e_2$ a `bad' state in the following sense:  in state $e_1$ the intensity of an upward move  of the stock is larger than  in state $e_2$;   the intensity   of a downward move on the other hand is larger in state $e_2$ than in $e_1$.
We therefore choose constants  $\cupp_1 > \cupp_2 >0$,  $\cdown_2 > \cdown_1>0$ and a price impact parameter  $a> 0$ and we set for $i=1,2$,
\begin{equation}\label{eq:lambda-ex2}
\lambda^+ (e_i, \nu )  = (\cupp_1, \cupp_2) e_i \; \text{ and } \lambda^- (e_i, \nu ) = (1 + a \nu) (\cdown_1, \cdown_2) e_i.
\end{equation}
 Then,
$
\eta^\P( e_i, \nu,\ud z) = \lambda^+ (e_i, \nu ) \delta_{\{\theta \}}(\ud z) +
 \lambda^- (e_i, \nu)  \delta_{ \{-\theta \}}(\ud z)
$, for $i=1,2$.
Since $\cupp_1 > \cupp_2$, in state $e_1$ one has on average more  buy orders; this   might represent a scenario where   another trader is executing a large buy program.   Similarly, since $\cdown_2 > \cdown_1$,  there are on average more sell orders  in state $e_2$, for instance because another trader is executing a large sell program. The form of $\eta^\P$ implies that the  permanent price impact is linear and  proportional to the intensity of a downward move and hence larger in the `bad' state $e_2$ than in the good state $e_1$.

Note that within our general setup this example could be enhanced in a number of ways. For instance, the transition intensities $ \cupp_i$ and $\cdown_i$ and the liquidity parameter $a$ could be made time dependent to reflect the fact that on most markets trading activity during the day is $U$-shaped with more trades occurring at the beginning and the end of a day than in the middle. Moreover, one  could introduce an additional state where the market is moving sideways, or one could consider the case where the liquidity parameter $a$ depends on $Y$.
\end{example}

\subsection{Bounds on the value function} \label{subsec:bounds}

In the previous section we have seen that having a temporary price impact described by a superlinear and convex function $\nu f(\nu)$ implies an endogenous upper bound $\numax$ for the liquidation strategy. However it is difficult to estimate this value if the exact form of the function $f$ is unknown. In this section we provide a robustness result by showing that the optimal proceeds from liquidation are almost independent of the precise value of $\numax$.  To this, we define $J^{*,m}$ as the  optimal liquidation value if the trader uses $\bF^S$-adapted strategies with $\nu_t \le m$ for all $t$ and prove in Proposition~\ref{prop:bound-on-value} that $J^{*,m}$ is bounded independently of $m$. Now the sequence $\{J^{*,m}\}_{m \in \bN}$ is obviously  increasing,  since a  higher $m$ means that the trader can optimize over a larger set of strategies. Hence, $\{J^{*,m}\}_{m \in \bN}$ is Cauchy which leads to the result.

\begin{proposition}\label{prop:bound-on-value}
Suppose that Assumption~\ref{ass1} holds and that the function $(t,e,\nu)\to\overline{\eta}^\P(t,e, \nu)$ from \eqref{eq:drift} is decreasing in $\nu$, and set
$$\overline{\eta} = 0 \vee \sup \{\overline{\eta}^\P(t,e, 0) - \rho \colon t \in[0,T], e \in \mathcal{E}\}.$$
Then $\sup_{m >0} J^{*,m} \le w_0 S_0 e^{ \overline{\eta}T}$.
\end{proposition}
Note that the upper bound on $J^*$ corresponds to the liquidation value of the inventory in a frictionless model  where the expected value of the bid price grows at the maximum rate $ \overline{\eta} + \rho$.

\begin{proof} Fix some $\bF^S$-adapted strategy $\Nu$ with values in  $[0,m]$, denote the corresponding  bid price  by $ S^{\Nu}$ and let $\widetilde{S}^{\Nu}_t = e^{-\rho t} S_t^{\Nu} $. Since   $W_t = w_0 - \int_0^t \nu_s \ud s$ we get by partial integration that
$$
\int_0^\tau \nu_s \widetilde{S}_s^{\Nu} \ud s = - \int_0^\tau \widetilde{S}_s^{\Nu} \ud W_s =
S_0 w_0 - \widetilde{S}_\tau^{\Nu}  W_\tau + \int_0^\tau W_s \ud \widetilde{S}^{\Nu}_s\,.
$$
Since $h(w) \le w$ and $f(\nu) \ge 0$ we thus get that
$$
\int_0^\tau  \! \nu_u \widetilde{S}_u^{\Nu} (1 -f(\nu_{u})) \ud u +  \widetilde{S}_\tau^{\Nu} h( W_\tau)
\le \int_0^\tau \nu_u \widetilde{S}_u^{\Nu} \ud u  + \widetilde{S}_\tau^{\Nu} W_\tau  =
S_0 w_0 + \int_0^\tau W_u \ud \widetilde{S}^{\Nu}_u .
$$
Notice that $ \int_0^\tau W_u \ud \widetilde{S}^{\Nu}_u = \int_0^\tau W_u  \widetilde{S}^{\Nu}_u \ud M_u^R +
\int_0^\tau W_u  \widetilde{S}^{\Nu}_u ( \overline{\eta}^\P(u,Y_{u-}, \nu_{u-}) -\rho) \ud u$.
Moreover the process$\int_0^{\cdot \wedge \tau} W_u  \widetilde{S}^{\Nu}_u \ud M_u^R$ is a true martingale.  As $0 \le W_u \le w_0$,  a similar argument as in the proof of  Lemma~\ref{lemma:martingale-part-of-S} shows that this process  is of integrable quadratic variation.   Since   $\overline{\eta}^\P(u,Y_{u-}, \nu_{u-}) -\rho \le  \overline{\eta}$,  $\tau  \le T$ and $W_u \le w_0$, we  get
\begin{align}
J(\Nu) &\le S_0 w_0 + \mathbb{E} \Big ( \int_0^\tau  W_u  \widetilde{S}^{\Nu}_u  (\overline{\eta}^\P(u,Y_{u-}, \nu_{u-}) - \rho)
\ud u \Big )
 \\ \label{eq:estimate-on-J}
 & \le S_0 w_0 + \mathbb{E} \Big ( \int_0^T w_0  \widetilde{S}^{\Nu}_u \overline{\eta}\,  \ud u \Big)\,.
\end{align}
Next we show that $\mathbb{E} \big( \widetilde{S}_t^{\Nu} \big) \le S_0 e^{\overline{\eta} t}  $. To this end, note that by Lemma \ref{lemma:martingale-part-of-S}, $\int_0^\cdot S_{s-}^{\Nu} \ud M_s^R$ is a true martingale so that
$$ \mathbb{E} \big( \widetilde{S}_t^{\Nu} \big) = S_0 +
   \mathbb{E} \Big ( \int_0^t \widetilde{S}^{\Nu}_u  (\overline{\eta}^\P(u,Y_{u-}, \nu_{u-}) - \rho) \ud u \Big )
   \le S_0 + \overline{\eta} \int_0^t \mathbb{E} \big( \widetilde{S}_u^{\Nu} \big) \ud u,
$$
and the claim follows from the Gronwall inequality. Using \eqref{eq:estimate-on-J} we finally get  that
$
J(\Nu) \le S_0 w_0 (1 + \int_0^T \overline{\eta} e^{\overline{\eta} u} \ud u ) = S_0 w_0 e^{\overline{\eta} T}
$,
and hence the result.
\end{proof}

\section{Partial Information and Filtering}\label{sec:partial info}

In this section we derive the filtering equations for our model.
Filtering for point process observations is for instance considered in \citet{frey2012pricing, ceci2012nonlinear, ceci2014zakai}. This literature is mostly based on the innovations approach. In this paper, instead, we address  the filtering problem  via the {\em reference probability approach}. This methodology relies on the existence of an equivalent probability measure such that the observation process is driven by a random measure with dual predictable projection independent of the Markov chain, see for instance \citet[Chapter 6]{bremaud1981point}.  The reference probability approach permits us to give a rigorous construction of our model, see Lemma~\ref{lem:change-of-measure}.


\subsection{Reference probability.}\label{sec:ref_probability}

We start from a filtered probability space $(\Omega,\mathcal{F},\mathbb{F},\Q)$ that supports a Markov chain $Y$ with state space $\mathcal{E}$ and generator matrix  $Q$, and an \emph{independent} Poisson random measure $\mu^R$  with compensator $\eta^\Q(\ud z) \ud t$ as in Assumption \ref{ass1}.2; $\Q$ is known as the {\em reference probability} measure.
Note that  the independence of $Y$ and $\mu^R$ implies that  $R$ and $Y$ have no common jumps.
For  $(t,e,\nu,z) \in [0,T]\times \mathcal{E}\times [0,\numax] \times \text{supp}(\eta) $, we define the function $\beta$ by
\begin{equation} \label{eq:def-beta}
 \beta(t,e, \nu, z): =\frac{\ud \eta^\P (t,e, \nu;\ud z) }{\ud \eta^\Q(\ud z)}(z) - 1\,,
\end{equation}
i.e. $\beta(t,e, \nu, z)+1$ is the Radon-Nikodym derivative of the measure $\eta^\P (t,e, \nu; \ud z)$  with respect to $\eta^\Q(\ud z)$.

We denote by $\bF^R$ the filtration generated by $\mu^R$.   Fix some $\bF^R$-adapted  liquidation strategy $\Nu$ with $\nu_t \in [0, \numax]$, $t \le T$ and define for $t \in[0,T]$ the stochastic exponential  $\widetilde Z$ by
\begin{equation} \label{eq:densityZtilde}
\widetilde{Z}_{t}=1+ \int_0^t \int_{\mathbb{R}} \widetilde{Z}_{s^-} \beta( s,Y_{s^-},\nu_{s^-}, z) \left( \mu^R(\ud s,\ud z) -\eta^\Q(\ud z) \ud s  \right).
\end{equation}
Then we have the following result.
\begin{lemma}\label{lem:change-of-measure} Let Assumption \ref{ass1} prevail. Then the process $\widetilde Z$ is a strictly positive  martingale with $\bE^\Q \big( \widetilde Z_T\big) =1.$ Define a measure $\P$ on $\mathcal{F}_T$ by setting $\frac{\ud \P}{\ud \Q}\big |_{\F_T} = \widetilde Z_T$. Then $\P$ and $\Q$ are equivalent and, under $\P$, the random measure $\mu^{R}$ has the compensator $\eta^\P$.
\end{lemma}
The proof of the lemma is postponed to Appendix \ref{app:filtering}.

Note that Lemma~\ref{lem:change-of-measure} gives a rigorous construction of the model introduced in  Section~\ref{subsec:bid-price}. The advantage of using a change-of-measure approach is the fact that the Poisson random measure $\mu^R$ and the observation filtration $\bF^R$ are  exogenously given.  If one attempts a direct construction  circularities arise, as  the  process $R$   depends on the  strategy $\Nu$ which is in turn adapted to the filtration $\bF^R$.


\subsection{Filtering equations.}

For a function $f\colon\mathcal{E} \to \R$, we introduce the  filter $\pi(f)$ as the optional projection of the process $f(Y)$ on the filtration $\bF^S$, i.e.~$\pi(f)$ is a c\`{a}dl\`{a}g process such that for all $t \in [0,T]$, it holds that
$\pi_t(f)=\esp{f(Y_t)\mid\F^S_t}$. Note that since $Y$ is a finite state Markov chain  $f(Y_t) = \langle\mathbf{f},Y_t \rangle$ for all $t \in [0,T]$, where $\langle  \ , \ \rangle$ denotes the scalar product on $\R^K$ and $\mathbf{f}_i = f(e_i)$, $i \in \{1, \dots, K\}$, and therefore functions of the Markov chain can be identified with $K$-vectors.
Let for all $t \in [0,T]$ and $i \in \{1, \dots, K\}$, $\pi^i_t:=\esp{\I_{\{Y_t=e_i\}}\mid\F^S_t}$ . Then, we can represent the filter as
\begin{equation}\label{eq:filterMC}
\pi_{t}(f)=\sum_{i=1}^K \mathbf{f}_i\pi^i_t=\langle \mathbf{f},\pi_t \rangle, \quad 0 \leq t \leq T.
\end{equation}
The objective of this section is to  derive the dynamics of the process $\bpi = (\pi^1, \dots, \pi^K)$.
To this end, we first observe that by the  Kallianpur-Striebel formula we have
$\displaystyle \pi_t(f):=\frac{p_t(f)}{p_t(1)}$ for all $t \in [0,T]$,
where $p(f)$ denotes the \emph{unnormalized} version of the filter, which is defined by
\begin{equation}\label{eq:unnormalized}
p_t(f):=\qesp{\widetilde{Z}_t \langle \mathbf{f},Y_t \rangle  \mid\F^S_t}, \quad 0 \leq t \leq T.
\end{equation}
The dynamics of $p(f)$ is given in the next theorem.

\begin{theorem}[The Zakai equation]\label{thm:Zakai}
Suppose Assumption~\ref{ass1}  holds and let $f\colon\mathcal{E}\to \mathbb{R}$. Then, for all $t \in [0,T]$, the unnormalized filter \eqref{eq:unnormalized} solves the equation:
\begin{equation}\label{eq:zakai}
p_t(f)=\pi_0(f)+\int_0^tp_s(Qf)\ud s+\int_0^t\int_\R p_{s^-}\left( \beta(z) f\right)(\mu^R(\ud s, \ud z)-\eta^\Q_s(\ud z)\ud s),
\end{equation}
where  $p_{t^-}\left( \beta(z) f\right)=\qesp{f(Y_{t^-}) \widetilde Z_{t^-} \beta(t,Y_{t^-},\nu_{t^-},z)\mid\F^S_t}$ and $p_t(Qf)=\qesp{\widetilde{Z}_t \langle Q \mathbf{f},Y_t \rangle  \mid\F^S_t}$.
\end{theorem}

We now provide the general idea of the proof, details are given in  Appendix~\ref{app:filtering}. Consider the process $\widetilde Z$ defined in \eqref{eq:densityZtilde} and some function $f\colon\mathcal E\to \R$. Then by It\^o's formula the product $\widetilde Z_t f(Y_t)$ has the following $(\Q,\bF)$-semimartingale decomposition
\begin{align}
\widetilde Z_t f(Y_t)= &f(Y_0)+ \int_0^t \widetilde Z_s \langle Q\mathbf{f},Y_t \rangle\ud s+\int_0^t\widetilde Z_s \ud M^f_s\\
& + \int_0^t\widetilde Z_s f(Y_s)\int_\R \beta(s, Y_{s^-}, \nu_{s^-},z)\left(\mu^R(\ud s, \ud z)-\eta^\Q_{s}(\ud z)\ud s\right),\label{eq:semimg}
\end{align}
where $M^f=(M^f)_{0 \leq t \leq T}$ is the true $(\bF, \Q)$-martingale appearing in the semimartingale decomposition of $f(Y)$.
Taking the conditional expectation with respect to $\F^S_t$ yields the result, since it can be shown that  $\bE^\Q \big ( \int_0^t\widetilde Z_s \ud M^f_s \mid \F_t^S \big ) =0$.

We introduce the notation
\begin{equation}\label{eq:compensator}
\pi_{t^-}(\eta^\P(\ud z)):=\sum_{i=1}^K\pi^i_{t^-}\eta^\P(t,e_i,\nu_t,\ud z), \quad 0 \leq t \leq T.\end{equation}
By applying \cite[Ch. II, Theorem T14]{bremaud1981point} it is easy to see that $\pi_{t^-}(\eta^\P(\ud z))\ud t$ gives the $(\bF^S, \P)$-dual predictable projection of the measure $\mu^R$.
The next proposition provides the dynamics of the conditional state probabilities.

\begin{proposition}\label{cor:KS-equation}
The process $\bpi=(\pi^1, \dots, \pi^K) $ solves the following system of equations:
\begin{equation}\label{eq:KS}
\pi_t^i=\pi_0^i+\int_0^t\sum_{j=1}^Kq^{ji}\pi_{s}^j\ud s+\int_0^t\int_{\mathbb{R}} \pi_{s^-}^i u^i(s,\nu_{s^-},\pi_s, z)(\mu^R(\ud s, \ud z)-\pi_{s^-}(\eta^{\P}(\ud z))\ud s),
\end{equation}
for every $t \in [0,T]$ and $ 1 \le i \le K$, where $\ds u^i(t,\nu,\pi,z):=\frac{(\ud \eta^{\P}(t,e_i,\nu)/\ud \eta^{\Q})(z)}{\sum_{j=1}^K \pi^j ({\ud \eta^{\P}(t,e_j,\nu)}/{\ud \eta^\Q} )( z)}-1$.
\end{proposition}
\begin{proof}
By the Kallianpur-Striebel formula we have that $\displaystyle \pi_t(f):=\frac{p_t(f)}{p_t(1)}$, for every $t \in [0,T]$. Then, by \eqref{eq:zakai} and It\^{o} formula  we get the dynamics of the normalized filter $\pi(f)$. The claimed result is obtained by setting $f(Y_t)=\I_{\{Y_t=e_i \}}$, for every $i \in \{1, \dots, K\}$.
\end{proof}
Note that the filtering equation~\eqref{eq:KS} does not depend on the particular choice of $\eta^\Q$.

\paragraph{Filter equations for \protect{Example~\ref{exm2}}.} In the following we give the dynamics of the process $\bpi$ for Example~\ref{exm2}. For a two-state Markov chain it is sufficient to specify the dynamics of $\pi = \pi^1$, since $\pi^2=1-\pi^1$.
Define two point processes$ N_t^{\text{up}} = \sum_{T_n \le t} 1_{\{\Delta R_{T_n} = \theta\}}$ and $ N_t^{\text{down}} = \sum_{T_n \le t} 1_{\{\Delta R_{T_n} = - \theta\}} $,  for all $t \in [0,T]$, that count the upward and the downward jumps of the return process. It is easily seen  that for every $(\nu, \pi, z)\in [0, \numax]\times [0,1]\times \{-\theta, \theta\}$, the function $u^1$ is  given by
 \begin{equation}\label{eq:u1-example}
 u^1 (\nu, \pi, z ) = \frac{ \lambda^+(e_1, \nu)}{\pi \lambda^+ (e_1, \nu) + (1-\pi) \lambda^+ (e_2, \nu)} 1_{\{z = \theta\}}
                    + \frac{\lambda^-(e_1, \nu)}{\pi \lambda^- (e_1, \nu) + (1-\pi) \lambda^- (e_2, \nu)} 1_{\{z = - \theta\}}.
 \end{equation}
By Corollary~\ref{cor:KS-equation} we  get the following equation for $\pi_t = \pi_t^1$:
\begin{align}
\ud \pi_t &=  \left(q^{11}\pi_t + q^{21}(1-\pi_t)\right)\ud t\\
&+\pi_t(1-\pi_t)\left ((\lambda^+(e_1, \nu_t)+\lambda^-(e_1, \nu_t))
- (\lambda^+(e_2, \nu_t)+\lambda^-(e_2, \nu_t))\right ) \ud t\\
&+ \pi_{t-} \Big(\frac{\lambda^+(e_1, \nu_t)}{\pi_{t-}\lambda^+(e_1, \nu_t)+(1-\pi_{t-})\lambda^+(e_2, \nu_t)}-1\Big)\ud N_t^{\text{up}}\\
&+ \pi_{t-} \Big(\frac{\lambda^-( e_1, \nu_t)}{\pi_{t-}\lambda^-(e_1, \nu_t)+(1-\pi_{t-})\lambda^-(e_2, \nu_t)}-1\Big) \ud N_t^{\text{down}}.\label{eq:dynamics-pi-Ex2}
\end{align}

\section{Control Problem I: Analysis via  PDMPs}\label{section:mdm}

We begin with a brief overview of our analysis of the control problem~\eqref{eq:optimization_full}. In Proposition~\ref{prop:pdmp-exists} below we  show that the  Kushner-Stratonovich equation \eqref{eq:KS} has a unique solution. Then standard arguments  ensure that the original control problem under incomplete information  is equivalent to a control problem under complete information with state process equal to the $(K+2)$-dimensional process $X:=(W,S,\bpi)$. This process is a PDMP in the sense of \citet{davis1993markov}, that is a trajectory of  $X$  consists of a deterministic part which solves an ordinary differential equation (ODE),  interspersed by random jumps. Therefore, to solve the optimal liquidation problem we apply control theory for PDMPs. This theory is based on the observation that a control problem for a PDMP is discrete in time: loosely speaking,  at every jump-time of the process  one chooses a control policy to be followed up to the next jump time or until maturity. Therefore,  one can identify the control problem for the PDMP with a control problem for a discrete-time, infinite-horizon Markov decision model (MDM). Using this connection we show that the value function of the  optimal liquidation  problem is continuous and that is the unique solution of the dynamic programming or optimality equation for the MDM.  These results are the basis for the viscosity-solution characterization of the value function in Section~\ref{sec:viscosity}.



\subsection{Optimal liquidation as a control problem for a PDMP.}\label{subsec:PDMP}

From the viewpoint of the trader   endowed with the filtration $\bF^S$, the state of the economic system at time $t\in [0,T]$ is  given by  $X_t=(W_t,S_t,\pi_t)$. Since it is more convenient to work with autonomous Markov  processes we  include  time into the state and define $\widetilde X_t:=(t,X_t)$.
The  state space of $\widetilde{X}$ is $\widetilde{\mathcal X}=[0,T]\times \mathcal X$ where  $ \mathcal X=[0, w_0]\times\R^+\times \S^K $ with $\S^K$  being the $K$-dimensional simplex.
Let $\Nu $ be the liquidation strategy followed by the trader. It follows from \eqref{eq:inventory},  \eqref{eq:KS}, and from the fact that the bid price is a pure jump process  that between jump times  the state process follows the ODE $\ud \widetilde X_t=  g (\widetilde X_t,\nu_t)\ud t $, where   the   vector field $ g (\widetilde x,\nu)\in \mathbb{R}^{K+3}$  is given by
$ g^1 (\widetilde x,\nu)=1$, $ g^2 (\widetilde x,\nu)=-\nu$, $g^3 (\widetilde x,\nu)=0$,
and for $k=1,\dots,K$,
\begin{gather} \label{eq:def-g}
g^{k+3} (\widetilde x,\nu)=\sum_{j=1}^K q^{jk}\pi^j-\pi^k\sum_{j=1}^K \pi^j \int_{\mathbb{R}}  u^k(t,\nu,\pi,z)\eta^\P(t,e_j,\nu,\ud z).
\end{gather}
For our analysis we need the following regularity property of $g$.
\begin{lemma}\label{lemma:lipschitz}
Under Assumption~\ref{ass1}, the function $g$ is Lipschitz continuous in $\widetilde x$ uniformly in $(t,\nu)  \in [0,T] \times [0,\numax] $; the Lipschitz constant is denoted by $K_g$.
\end{lemma}
The proof is postponed to Appendix~\ref{app:sec-4}.

The jump rate of the state process  $\widetilde{X}$ is given by $\lambda(\widetilde X_{t^-}, \nu_{t^-})$, for all $t \in (0,T]$, where for every $(\widetilde x, \nu)\in \widetilde{\mathcal{X}}\times [0, \numax]$,
$$
\lambda(\widetilde x,\nu)=\lambda(t,w,s,\pi,\nu):=\sum_{j=1}^K \pi^j \eta^\P(t,e_j,\nu,\mathbb{R}).
$$
Next, we identify the  transition kernel $Q_{\widetilde X}$ that governs the jumps of  $\widetilde{X}$. Denote by $\{T_n\}_{n \in \bN}$ the sequence of jump times of $\widetilde{X}$. It  follows from \eqref{eq:KS} that for any measurable function $f \colon \widetilde{\mathcal{X}} \to \R^+ $,
\begin{align} \label{eq:def-Q}
 Q_{\widetilde X} f (\widetilde x, \nu )  &:= E\big ( f(\widetilde{X}_{T_n}) \mid T_n = t, X_{T_n -} =x, \nu_{T_n -} = \nu \big )  = \frac{1}{\lambda(\widetilde x,\nu)} \overline{Q}_{\widetilde X} f(\widetilde x, \nu) \,,
\end{align}
where the unnormalized kernel $\overline{ Q}_{\widetilde X}$ is given by
\begin{align}
 \overline{ Q}_{\widetilde X} f (\widetilde x, \nu ) =\sum_{j=1}^K \pi^j \int_{\mathbb{R}} f \big ( t,w,s(1+z),\pi^{1}(1+u^{1}),\dots,
 \pi^{K}(1+u^{K}) \big) \eta^\P(t,e_j,\nu,\ud z).
 \label{eq:def-Qbar}
\end{align}
Here $u^i$ is short for $u^i (t,\nu,\pi,z)$.
Summarizing, $\widetilde{X}$ is a PDMP with characteristics given by the vector field $g$, the jump rate $\lambda$ and the transition kernel $Q_{\widetilde X}$.

It is standard in control theory for PDMPs to work with so-called open-loop controls. In the current context this means that the trader chooses at each jump time $T_n<\tau$ a liquidation policy $\nu^n$ to be followed up to $T_{n+1}\wedge \tau$. This policy may depend on the state $\widetilde{X}_{T_n}=(T_n,X_{T_n})$.
\begin{definition} \label{def:admissible-control} Denote by $\mathcal A$ the set of measurable mappings $\alpha\colon[0,T]\to [0,\numax]$.  An \emph{admissible open loop liquidation strategy} is a sequence of mappings $\{\nu^n\}_{n \in \bN}$ with $\nu^n:\widetilde{\mathcal X}\to \mathcal A$;  the liquidation rate at time $t$ is given by
$\nu_t=\sum_{n=0}^\infty \I_{(T_n\wedge \tau, T_{n+1}\wedge \tau]} (t)   \nu^n (t- T_n,\widetilde X_{T_n}).$
\end{definition}
It follows from \citet[Theorem~T34, Appendix~A2]{bremaud1981point} that  an  admissible strategy is of the form given in Definition~\ref{def:admissible-control}, but for $\F^S_{T_n}$ measurable mappings $\nu^n \colon \Omega \to \mathcal A$ for every $n \in \bN$, that $\nu^n $  may depend on the entire history of the system.  General  results for Markov decision models (see \citet[Theorem~2.2.3]{bauerle-rieder-book})  show that the expected profit of the trader stays the same if instead we consider the smaller class of admissible open loop strategies, so that we may restrict ourselves to this class.

\begin{proposition}\label{prop:pdmp-exists}
Let Assumption~\ref{ass1} hold.
For every admissible liquidation strategy $\{\nu^n\}_{n \in \bN}$ and every initial value $\widetilde x$, a unique PDMP with characteristics $g$, $\lambda$, and $Q_{\widetilde X}$ as above exists.
In particular the Kushner-Stratonovic equation \eqref{eq:KS} has a unique solution.
\end{proposition}
\begin{proof}
Lemma \ref{lemma:lipschitz} implies that for $\alpha\in \A$ the ODE $\ud \widetilde X_t= g (\widetilde X_t,\alpha_t)\ud t $ has a unique solution so that between jumps the state process is well-defined.
At any jump time $T_n$, $\widetilde X_{T_n}$ is uniquely defined in terms of observable data $(T_n, \Delta R_{T_n})$.
Moreover, since the jump intensity is bounded by $\lambda^\text{max}$, jump times cannot accumulate.
\end{proof}

Denote by $\P^{\{\nu^n\}}_{(t,x)}$ (equiv. $\P^{\{\nu^n\}}_{\widetilde x}$) the law of the state process provided that  $X_t=x\in\mathcal{X}$ and that the trader uses the open-loop strategy $\{\nu^n\}_{n \in \bN}$. The reward function associated to an admissible liquidation strategy $\{\nu^n\}_{n \in \bN} $ is defined by
\begin{gather}
V\left(t,x,\{\nu^n\}_{n \in \bN} \right)=\mathbb E_{(t,x)}^{\{\nu^n\}}\Big( \int_t^{\tau} e^{-\rho (u-t)} \nu_u S_u(1-f(\nu_u)) \ud u + e^{\rho(\tau-t)}h\left(W_{\tau}\right)S_{\tau}    \Big ),
\end{gather}
and the value function of the liquidation problem under partial information  is
\begin{align}\label{1.7}
&V(t,x) =\sup \left\{ V\left(t,x,\{\nu^n\}_{n \in \bN} \right)\colon   \{\nu^n\}_{n \in \bN} \  \text{admissible liquidation strategy} \right \}. \quad{}
\end{align}

\begin{remark}\label{rem:V-pos-hom}
Note that the compensator $\eta^\P$ and the dynamics of the filter $\bpi$ are independent of the current bid price $s$,
and that the payoff of a liquidation strategy $\{\nu^n\}_{n \in \bN}$ is positively homogeneous in $s$. This implies that  the reward  and the value function of the liquidation problem are positively homogeneous in $s$ and, in particular,
$ V(t,w,s,\pi) = s V(t,w,1,\pi).$
\end{remark}

\subsection{Associated Markov decision model.}\label{sec:MDM}

The optimization problem in \eqref{1.7} is  discrete in time since  the control policy is chosen at the discrete time points $T_n$, $n \in \bN$ ,  and  the value of the state process at these time points forms a discrete-time Markov chain (for $T_n  < \tau$).  Hence \eqref{1.7} can be rewritten as a control problem in an infinite horizon Markov decision model.  The \emph{state process} of the MDM is given by the sequence $\{L_n\}_{n \in \bN} $ of random variables with
\begin{gather}
L_n=\widetilde X_{T_n}\ \text{ for } \ T_n<\tau \ \text{ and } \ L_n=\bar \Delta \  \text{ for } \ T_n\ge \tau, \quad n\in\bN\,,
\end{gather}
where $\bar \Delta$ is the cemetery state.
In order to derive the transition kernel of the sequence $\{L_n\}_{n \in \bN} $ and the  reward function of the MDM, we introduce some notation. For a function $\alpha\in\mathcal A$ we denote by $\widetilde \varphi^{\alpha}_t(\widetilde x)$ or by  $\widetilde \varphi_t (\alpha, \widetilde x)$ the flow of the initial value problem
$\frac{\ud }{\ud t}  \widetilde X(t)=g \big(\widetilde X(t), \alpha_t\big)$ with initial condition  $ \widetilde X(0) =\widetilde x.$
Whenever we want to make the dependence on time explicit we write  $\widetilde \varphi^{\alpha}$ in the form $(t,\varphi^{\alpha})$.  Moreover, we define the function $\lambda^{\alpha}_u$ by
\begin{equation}\label{eq:def-survival}
   \lambda^{\alpha}_u(\widetilde x) =\lambda(\widetilde \varphi^{\alpha}_u(\widetilde x),\alpha_u)=\lambda((t+u,\varphi^{\alpha}_u),\alpha_u)\, \quad u \in [0,T-t],
\end{equation}
and we let  $\Lambda^{\alpha}_u(\widetilde x) = \int_0^u \lambda^{\alpha}_v(\widetilde x)\ud v$.

Next we take a closer look at the \emph{boundary} of $\widetilde{X}$. First note that the process $\bpi$ takes values in  the hyperplane $ \H^K = \{x \in \R^K \colon \sum_{i=1}^K x_i = 1 \}$, so that  $\widetilde{\mathcal{X}}$ is contained in the set $\H = \R^3 \times \H^K$, which is a hyperplane of $\R^{K+3}$. When considering the boundary  or the interior of the state space  we always refer to the \emph{relative} boundary or the \emph{relative} interior with respect to $\H$.
 Of particular interest to us is the \emph{active boundary} $\Gamma$ of the state space, that is the part of the boundary of $\widetilde{\mathcal{X}}$ which can be reached by the flow $\widetilde \varphi^{\alpha}_\cdot(\widetilde x)$ starting in an interior point $\tilde x\in  \operatorname{int}(\widetilde{\mathcal{X}})$.
The boundary of  $\widetilde {\mathcal X}$ can only be reached if $w=0$, if  $t=T$, or if the filter process reaches the boundary of the $K$-dimensional simplex. The latter is not possible: indeed, if $\pi_0^i>0$, then $\pi_t^i>0$ for all $t\in [0,T]$, since  there is a positive probability that the Markov chain has not changed its state and since the conditional distribution of $Y_t$ given $\F_t^S$ is equivalent to the unconditional distribution of $Y_t$ by the Kallianpur-Striebel formula.
Hence the active boundary equals $\Gamma = \Gamma_1 \cup \Gamma_2$,  where
\begin{equation} \label{eq:active-boundary}
\Gamma_1 = [0,T] \times \{0\} \times (0, \infty) \times \S^K_0 \ \text{ and } \ \Gamma_2 = \{T\} \times [0, w_0] \times (0, \infty) \times \S^K_0,
\end{equation}
and where $\S^K_0$ is the interior of $\S^K$, i.e. $\S^K_0:=\{x\in \S^K \colon x_i > 0  \text{ for all } i\}$.
In \eqref{eq:active-boundary} $\Gamma_1$ is the \emph{lateral} part of the active boundary corresponding  to an inventory level equal to zero, and $\Gamma_2$  is the terminal boundary corresponding to  the exit from the state space at maturity $T$.
In the sequel we denote  the first exit time of the flow $\widetilde \varphi^{\alpha}_\cdot (\widetilde x)$  from  $\widetilde{\mathcal{X}}$ by
\begin{gather}
\tau^{\varphi}=\tau^{\varphi}(\widetilde x, \alpha)=\inf\{u\ge 0: \widetilde\varphi^{\alpha}_u(\widetilde x) \in \Gamma \}\,.
\end{gather}
Notice that the stopping time $\tau$ defined in \eqref{tau1} corresponds to the first time the state process $\widetilde X$ reaches the active boundary $\Gamma$.

Using similar arguments as in \citet[Section~8.2]{bauerle-rieder-book} or in \citet[Section~44]{davis1993markov}, it is easily seen that the \emph{transition kernel} $Q_L$ of the sequence  $\{L_n\}_{n \in \bN} $ is given by
\begin{equation}\label{eq:form-of-QL}
Q_L f \big( (t,x), \alpha\big ) =  \int_0^{ \tau^\varphi(\widetilde x)} \! \!  e^{- \Lambda^\alpha_u (\widetilde x)} \overline{Q}_{\widetilde X} f(u+t, \varphi_u(\widetilde x), \alpha_u \big ) \ud u + e^{- \Lambda_{\tau^\varphi}^\alpha (\tilde x)} f(\bar \Delta) \,;
\end{equation}
we omit the details. Moreover, since the cemetery state is absorbing, $Q_L \I_{\{\bar \Delta\}} (\bar \Delta, \alpha)=1$.
Finally we  define the \emph{one-period reward function}  $r \colon \widetilde{\mathcal X}\times \mathcal A\to \mathbb{R}^+$ by
\begin{align}\label{rewardmdm}
r(\widetilde x, \alpha)=\int_0^{\tau^{\varphi}}e^{-\rho u}e^{-\Lambda^{\alpha}_u(\widetilde x)}\alpha_u s (1 - f(\alpha_u))\ud u + e^{-\rho \tau^\varphi} e^{-\Lambda^{\alpha}_{\tau^{\alpha}}(\widetilde x)}h(w_{\tau^{\varphi}})s,
\end{align}
and $w_{\tau^{\varphi}}$  the inventory-component of $\widetilde \varphi^{\alpha}$, and we set $r(\bar \Delta) =0$.
For an admissible  strategy $\{\nu^n\}_{n \in \bN} $ we set
$ J_{\infty}^{\{\nu^n\}}(\widetilde x)=\mathbb{E}_{\widetilde x}^{\{\nu^n\}}\Big( \sum_{n=0}^{\infty}r\left( L_n,\nu^n(L^n)\right) \Big)$, and
\begin{align}\label{2.8}
J_{\infty}(\widetilde x):=\sup \left\{J_{\infty}^{\{\nu^n\}}(\widetilde x):\{\nu^n\}_{n \in \bN} \ \text{admissible liquidation strategy}   \right\}.
\end{align}
The next lemma shows that the MDM with transition kernel $Q_L$ and one-period reward $r(L,\alpha)$ is equivalent to the optimization problem \eqref{1.7}.
\begin{lemma}\label{lemma2.1} For every admissible  strategy $\{\nu^n\}$ it holds that $V^{\{\nu^n\}}=J_{\infty}^{\{\nu^n\}}$.  Hence $V=J_{\infty}$, and  the control problems \eqref{1.7} and \eqref{2.8} are equivalent.
\end{lemma}
The proof is similar to the proof of \citet[Theorem~44.9]{davis1993markov} and is therefore omitted.

\subsection{The Bellman equation.}\label{subsec:optimality-equation}

In this section we study the Bellman  equation for the value function $V$. Define for  $\alpha\in\mathcal{A}$ and a measurable function $v\colon\widetilde{\mathcal{X}}\to \mathbb{R}^+$  the function $\mathcal{L}v(\cdot,\alpha)$ by
\begin{equation} \label{eq:def-L}
\mathcal{L}v( \widetilde{x},\alpha)=r(\widetilde{x},\alpha)+  Q_L v (\widetilde{x}, \alpha), \; \widetilde{x}  \in \widetilde{\mathcal{X}}.
\end{equation}
The \emph{maximal reward operator} $\mathcal{T}$ is then given by $\mathcal{T}v(\widetilde{x})=\sup_{\alpha\in\mathcal{A}}\mathcal{L}v(\widetilde{x},\alpha)$. Since the one-period reward function is nonnegative we have a so-called \emph{positive} MDM and it follows from  \citet[Theorem~7.4.3]{bauerle-rieder-book}  that the value function satisfies the so-called \emph{Bellman} or \emph{optimality equation}
\begin{equation} \label{eq:bellman-equation}
V(\widetilde x) = \mathcal{T} V(\widetilde x), \quad \widetilde{x}  \in \widetilde{\mathcal{X}},
\end{equation}
that is $V$ is a fixed point of the operator $\mathcal{T}$. In order to characterize $V$ as viscosity solution of the HJB equation associated with the PDMP $\widetilde{X}$ (see Section~\ref{sec:viscosity}) we need a stronger result. We want to show: i.) that the value function  $V$ is the unique fixed point of $\mathcal{T}$ in a suitable function class $\mathcal{M}$; ii.) that for a starting point $v^0 \in \mathcal{M} $ iterations of the form $v^{n+1} = \mathcal{T} v^n$, $n \in \bN 1$, converge to $V$; and iii.) that $V$ is continuous on $\widetilde{\mathcal{X}}$.

Points  i.) and ii.) follow from the next lemma.
\begin{lemma}\label{lemma:bounding-function} Define for  $\gamma >0$, the function  $b \colon
\widetilde{\mathcal{X}} \cup \{\bar \Delta\} \to \R^+$ by
$b(\widetilde{x})=b(t,x):=s w e^{\gamma(T-t)}$,  $\widetilde{x} \in \widetilde{\mathcal X}$, and $b(\bar \Delta) = 0$.
Then under  Assumption~\ref{ass1},    $ b$ is a \emph{bounding function} for the MDM with transition  kernel $Q_L$ and  reward function $r$, that is there are constants $c_r, c_Q $ such that for all $(\widetilde{x},\alpha)\in \widetilde{\mathcal{X}}\times\mathcal{A}$.
$$ |r(\widetilde{x},\alpha)| \le c_r b(\widetilde{x})  \  \text{ and } \ Q_L b (\widetilde{x}, \alpha) \le c_Q b(\widetilde{x})\,.$$
Moreover, for $\gamma$ sufficiently large it holds that $c_Q <1$, that is the MDM is \emph{contracting}.
\end{lemma}
The proof is postponed to Appendix \ref{app:sec-4}. In the sequel  we denote by $\mathcal{B}_b$ the set of functions
\[\mathcal{B}_b:=\big\{ v:\widetilde{\mathcal{X}}\to \mathbb{R} \text{ such that } \sup\nolimits_{\widetilde{x} \in\widetilde{\mathcal{X}}} \big |\nicefrac{  v(\widetilde{x}) }{ b(\widetilde{x})}\big | <\infty     \big\}\,,   \]
and we define for $v\in \mathcal{B}_b$ the norm   $\| v\|_b=\sup_{\widetilde{x}\in\widetilde{\mathcal{X}}}  |\nicefrac {v(\widetilde{x})}{ b(\widetilde{x})}|$. Then the following holds, see \citet[Section~ 7.3]{bauerle-rieder-book}: a) $(\mathcal{B}_b,\|\cdot\|)_b$ is a Banach space; b) $\mathcal{T} (\mathcal{B}_b) \subset \mathcal{B}_b$;
c) $\|\mathcal{T}v-\mathcal{T}u \|_b \le c_Q \|v-u \|_b$.

If the MDM is contracting,  the maximal reward operator is a contraction on $(\mathcal{B}_b,\|\cdot\|)_b$, and the value function is an element of $\mathcal{B}_b$.
Banach's fixed  point theorem thus gives properties  i.) and ii.) above  with $\mathcal{M} = \mathcal{B}_b$. In order to establish property iii.) (continuity of $V$) we observe that the set
\begin{equation}
\mathcal{C}_b := \{ v \in \mathcal{B}_b \colon \text{ $v$ is continuous}\}
\end{equation}
is a closed subset of $(\mathcal{B}_b,\|\cdot\|)_b$, see  \citet[Section~7.3]{bauerle-rieder-book}. Moreover, we show in Proposition~\ref{prop:continuity} that under certain continuity conditions (see Assumptions~\ref{ass1}~and~\ref{ass2}), $\mathcal{T} $ maps $\mathcal{C}_b$ into itself. Hence it follows from Banach's fixed point theorem that $V \in \mathcal{C}_b$.

\begin{assumption} \label{ass2}
\begin{itemize}
\item[1.] The measure   $\eta^j(t,\nu;\ud z)$ for $j \in \{1, \dots, K\}$ is continuous in  the weak topology, i.e.~for all bounded and continuous  $\phi$, the mapping $ (t, \nu) \mapsto \int_\R \phi(z) \eta^j(\ud z) $ is continuous on $[0,T] \times [\numin, \numax]$.

\item[2.] For the functions $u^j$  introduced in Proposition~\ref{cor:KS-equation} the following holds: for any sequence $\{(t^n, \nu^n, \pi^n)\}_{n \in \bN}$ with $ (t^n, \nu^n, \pi^n)\in [0,T) \times [\numin, \numax] \times \S^K$ for every $n \in \bN$, such that $\ds (t^n, \nu^n, \pi^n) \xrightarrow[n \to \infty]{} (t, \nu, \pi)$, one has
$$
\lim_{n \to \infty} \ \sup\nolimits_{z \in {\rm supp}(\eta)} |u^j(t^n, \nu^n, \pi^n,z) - u^j(t, \nu, \pi, z) | =0\,.
$$
\end{itemize}
\end{assumption}

\begin{proposition}\label{prop:continuity} Suppose that Assumptions~\ref{ass1}~and~\ref{ass2} hold  and let $v \in \mathcal{C}_b$. Then $\mathcal{T} v \in C_b$.
\end{proposition}

\begin{proof}  Consider some sequence $\widetilde{x}_n \to \widetilde{x}$ for $n \to \infty$. Since
$ | \cT v (\widetilde{x}_n) - \cT v (\widetilde{x}) | \le \sup_{\alpha \in \mathcal{A}} | \mathcal{L}v( \widetilde{x}_n,\alpha) -\mathcal{L}v( \widetilde{x},\alpha)|,
$
it suffices to estimate   the difference $ \sup_{\alpha \in \mathcal{A}} | \mathcal{L}v( \widetilde{x}_n,\alpha) -\mathcal{L}v( \widetilde{x},\alpha)|$.
  First, note that by the Lipschitz continuity of $g$, established in Lemma~\ref{lemma:lipschitz}, we have
$$\big | \widetilde\varphi_t^{\alpha} (\widetilde{x}_n) -  \widetilde \varphi_t^{\alpha} (\widetilde{x}) \big | \le
|\widetilde{x}_n - \widetilde{x}| + K_g  \int_0^t \big | \widetilde\varphi_u^{\alpha} (\widetilde{x}_n) -  \widetilde \varphi_u^{\alpha} (\widetilde{x}) \big | \,\ud u\,.$$
Gronwall inequality hence yields that
\begin{equation}\label{eq:convergence-of-flow}
\sup\nolimits_{t \in [0,T], \alpha \in \mathcal{A}} \big | \widetilde\varphi_t^{\alpha} (\widetilde{x}_n) -  \widetilde \varphi_t^{\alpha} (\widetilde{x}) \big |  \le |\widetilde{x}_n - \widetilde{x}| e^{K_g T}\,,
\end{equation}
and thus uniform convergence for $n \to \infty$ of the flow $\widetilde\varphi^{\alpha} (\widetilde{x}_n)$ to $\widetilde\varphi^{\alpha} (\widetilde{x})$.
This does however not imply that $\tau^{\varphi_n}$, the entrance time of $\widetilde \varphi^{\alpha} (\widetilde{x}_n)$  into the active boundary of the state space, converges to $\tau^\varphi$ for $n \to \infty$. To deal with this issue we distinguish two cases:

\textit{Case 1.} The flow $\widetilde \varphi_\cdot^{\alpha} (\widetilde{x})$ exits the state space $\widetilde{\mathcal{X}}$ at the terminal boundary $\Gamma_2$ (see \eqref{eq:active-boundary}). This implies that $\tau^\varphi = T-t$ and that the inventory level $w_u$ is strictly positive for $u < T-t$.  We therefore conclude from \eqref{eq:convergence-of-flow} that $\tau^{\varphi_n}$ converges to $T-t$.  Under Assumptions~\ref{ass1} and~\ref{ass2} the uniform convergence
$
\lim_{n \to \infty} \sup_{\alpha \in \mathcal{A}} | \mathcal{L}v( \widetilde{x}_n,\alpha) -\mathcal{L}v( \widetilde{x},\alpha)| =0
$
thus follows  immediately using the definition of $r$ and the continuity of the mapping $(\widetilde{x}, \nu) \mapsto \bar Q v(\widetilde x, \nu)$ established in Lemma~\ref{lemma:continuity-of-QL}, see Appendix~\ref{app:sec-4}.

\textit{Case 2.} The flow $\widetilde \varphi_\cdot^{\alpha} (\widetilde{x})$ exits $\widetilde{\mathcal{X}}$ at the lateral boundary $\Gamma_1$ so that $w_{\tau^\varphi} =0$. In that case \eqref{eq:convergence-of-flow}  implies that  $\liminf_{n \to \infty} \tau^{\varphi_n} \ge \tau^\varphi$; it is however possible that this inequality is strict.
We first show continuity of  the reward function for that case.   We decompose $r(\widetilde x_n,\alpha)$ as follows, setting $\rho =0$ for simplicity:
\begin{align} \label{eq:first-integral}
r(\widetilde x_n,\alpha)  & = s \int_{0}^{\tau^{\varphi}\wedge \tau^{\varphi_n} } \hspace{-0.2cm}  e^{-\Lambda^{\alpha}_u (\widetilde{x}_n)} \alpha_{u} (1 - f(\alpha_u)) \ud u  \\
&   + s \int^{\tau^{\varphi_n}}_{\tau^{\varphi}\wedge \tau^{\varphi_n} } \hspace{-0.2cm} e^{-\Lambda^{\alpha}_u (\widetilde{x_n})} \alpha_{u} (1 - f(\alpha_u)) \ud u +
 s  e^{-\Lambda^{\alpha}_{\tau^{\varphi_n}}(\widetilde{x}_n)}h(w_{\tau^{\varphi_n}}). \label{eq:second-integral}
\end{align}
Now it follows from  \eqref{eq:convergence-of-flow} that the integral in \eqref{eq:first-integral}  converges for $n \to \infty$ to $r(\widetilde x,\alpha)$ uniformly in $\alpha \in \mathcal{A}$.  The terms in \eqref{eq:second-integral} are bounded from above by $ s w_{\tau^{\varphi}\wedge \tau^{\varphi_n}}$; this can be shown via a similar partial integration argument as in the proof of Lemma~\ref{lemma:bounding-function}. Moreover,
$w_{\tau^{\varphi}\wedge \tau^{\varphi_n}}$ converges uniformly in $\alpha \in \mathcal{A}$ to $w_{\tau^\varphi}=0 $, so that \eqref{eq:second-integral} converges to zero.
Next  we  turn to the  transition kernel.  We decompose $Q_L v$:
\begin{align*}
Q_L v\big(\widetilde{x}_n,\alpha\big)  &= \int_0^{ \tau^\varphi\wedge \tau^{\varphi_n}} \hspace{-0.3cm} e^{- \Lambda_u^{\alpha} (\tilde x_n)}  \overline{Q} v(\widetilde \varphi_u(\widetilde x_n), \alpha_u)  \ud u
 + \int^{\tau^{\varphi_n}}_{\tau^{\varphi}\wedge \tau^{\varphi_n} } \hspace{-0.3cm} e^{- \Lambda_u^\alpha (\tilde x_n)}  \overline{Q} v(\widetilde \varphi_u^{\alpha}(\widetilde x_n), \alpha_u  \big ) \ud u\,.
\end{align*}
For $n \to \infty$, the first integral converges to $Q_L v\big(\widetilde{x},\alpha\big)$ using \eqref{eq:convergence-of-flow} and the continuity of  the mapping $(\widetilde{x}, \nu) \mapsto \overline{Q} v(\widetilde x, \nu)$ (Lemma~\ref{lemma:continuity-of-QL}). To estimate  the second term  note that
$\overline{Q} v(\widetilde x, \nu) \le \|v\|_b s w \lambda({\widetilde{x},\nu})$ (as $\frac{1}{\lambda}\bar Q$ is a probability transition kernel), so that the integral is bounded by
$$ \|v\|_b s w_{\tau^{\varphi}\wedge \tau^{\varphi_n}}  \int^{\tau^{\varphi_n}}_{\tau^{\varphi}\wedge \tau^{\varphi_n} } \hspace{-0.1cm} \lambda^{\alpha}_u e^{- \Lambda_u^\alpha (\widetilde{x}_n)} \ud u \le \|v\|_b s w_{\tau^{\varphi}\wedge \tau^{\varphi_n}}\,, $$
and the last term converges to zero for $n \to \infty$, uniformly in $\alpha \in \mathcal{A}$.
\end{proof}

\begin{remark}\label{rem:no-bound}
Note that  existing  continuity results for $\mathcal{L}v( \cdot,\alpha)$ such as \citet[Theorem~44.11]{davis1993markov}  make the assumption that the flow  $\varphi^{\alpha}$ reaches the active boundary at a uniform speed, independent of the chosen control. In order to ensure this hypothesis in our framework we would have to impose a strictly positive lower bound on the admissible liquidation rate. This is  an economically implausible restriction of the strategy space which is why we prefer to rely on a direct argument.
\end{remark}

We summarize the results of this section in the following theorem.

\begin{theorem}\label{thm:fixed_point} Suppose that Assumptions~\ref{ass1}~and~\ref{ass2} hold.  Then
the value function $V$ is continuous on $\widetilde{\mathcal X}$ and satisfies the boundary conditions  $V(\widetilde x) = 0 $ for $\widetilde x$ in the lateral boundary $\Gamma_1$ and $V(T, x) = s h(w)$. Moreover, $V$ is the unique solution of the Bellman or optimality equation $ V= \widetilde{\mathcal{T}} V$ in $\mathcal{B}_b$.
\end{theorem}

\section{Control Problem II: Viscosity Solutions }\label{sec:viscosity}

In this section we   show that  the value function is a  viscosity solution of the standard HJB equation associated with the controlled Markov process $(W, \pi)$ and  we derive a  comparison principle for that equation. These results are crucial to ensure  the convergence of suitable numerical schemes for the HJB equation and thus for the numerical solution of the optimal liquidation problem.
In Section~\ref{subsec:counterexample} we provide an example which shows that in general the HJB equation does not admit a classical solution.

\subsection{Viscosity solution characterization.}\label{subsec:characterization}

As a first step we write down the Bellman equation and we   use the positive homogeneity of $V$ in the bid price (see Remark~\ref{rem:V-pos-hom}) to eliminate $s$ from the set of state variables.  Define $ \widetilde{\mathcal{Y}} = [0,T] \times [0, w_0] \times \S^K $ and denote by $\operatorname{int} \widetilde{\mathcal{Y}}$ and
$\partial \widetilde{\mathcal{Y}}$  the \emph{relative} interior and the \emph{relative} boundary of $\widetilde{\mathcal{Y}}$ with respect to the hyperplane $\R^2 \times \H^K$. For $\widetilde y  \in \widetilde{\mathcal{Y}} $ we set
\begin{equation} \label{eq:def-V'}
 V' (\widetilde y) =  V'(t,w,\pi) := V(t, w,1, \pi)\,,
 \end{equation}
so that the value function satisfies the relation $V(\widetilde x) = s V'(\widetilde y)$. For $\nu \in [\numin, \numax]$, $\widetilde y  \in \widetilde{\mathcal{Y}}$, and any measurable function $\Psi \colon \widetilde{\mathcal{Y}} \to \R^+$, define
$$
\overline{ Q}' \Psi ( \widetilde y , \nu) := \sum_{j=1}^K \pi^j \int_{\mathbb{R}} (1+z) \Psi \left(t,w,(\pi^i(1+u^i(t,\pi,\nu,z)))_{i=1,\dots,K} \right) \eta^{\P}(t,e_j,\nu,\ud z)
$$
and note that $\overline{Q} V ( \widetilde x, \nu) =  s \overline{ Q}' V' (\widetilde y,\nu)$. From now on  we denote  by $\widetilde \varphi_u^\alpha (\widetilde y) $ the flow of the vector field $g$ with price component $g^3$ omitted, and we write $\tau^{\varphi}$ for the first time this flow reaches the active boundary of $\widetilde{\mathcal{Y}}$ given by $\Gamma := [0,T] \times \{0\} \times \S^K_0  \cup \{T\} \times [0, w_0]  \times \S^K_0\,$ of $\widetilde{\mathcal{Y}}$.

By positive homogeneity, the Bellman equation for $V$  reduces to  the following  optimality equation for $V^\prime$:
\begin{equation}\label{eq:optimality-for-Vprime}
\begin{split}
V^\prime (\widetilde y) &= \sup_{\alpha \in A}  \Big \{ \int_0^{\tau^{\varphi} } \hspace{-0.1cm}e^{-(\rho u +\Lambda^\alpha_u(\widetilde y))} \big (   \alpha_u (1 - f(\alpha_u)) +  \overline{ Q}^\prime   V^\prime  ( \widetilde{\varphi}_u^\alpha (\widetilde y)  , \alpha_u  ) \big ) \ud u \\
 & \quad + e^{- (\rho \tau^{\varphi}+ \Lambda_{\tau^{\varphi}}^\alpha (\widetilde y) )} h (w_{\tau^{\varphi}} )  \Big \}.
\end{split}
\end{equation}
For $\Psi \colon \widetilde{\mathcal{Y}} \to \R^+$ bounded, define the function $\ell^\Psi \colon \widetilde{\mathcal{Y}} \times [\numin, \numax] \to \R^+ $  and the operator $\mathcal{T}'$ by
\begin{align} \label{eq:def-fpsi}
 \ell^\Psi (\widetilde y, \nu) &= \nu (1- f(\nu)) + \overline{Q}' \Psi ( \widetilde y , \nu)\,, \\
\label{eq:def-Tprime}
\mathcal{T}^\prime \Psi (\widetilde y) &= \sup_{\alpha \in A} \Big \{ \int_0^{\tau^{\varphi} } \hspace{-0.2cm}e^{-(\rho u + \Lambda^\alpha_u(\widetilde y))}  \ell^{\Psi}\big(\widetilde \varphi_u^\alpha(\widetilde y), \alpha_u) \ud u
+ e^{- (\rho \tau^{\varphi}+ \Lambda_{\tau^{\varphi}}^\alpha (\widetilde y) )}h(w_{\tau^{\varphi}})  \Big \}.
\end{align}
Note that for fixed $\Psi$, $v^\Psi := \mathcal{T}^\prime \Psi$  is the value function of a deterministic {exit-time optimal control problem} with instantaneous  reward $\ell^\Psi$ and boundary value $h$.
Viscosity solutions  for this problem are studied extensively in \citet{bib:barles-94}. Moreover,
the optimality equation \eqref{eq:optimality-for-Vprime} for $V'$ can be written as the fixed point equation $V^\prime  = \mathcal{T}^\prime V^\prime$. \citet{bib:davis-farid-99} observed  that this can be used to obtain a viscosity solution characterization of the value function in a PDMP control problem,  and we now explain how this  idea applies  in our framework. Define for $\Psi \colon \widetilde{\mathcal Y} \to \R^+$ the function
$ F_\Psi \colon \widetilde{\mathcal Y} \times \R^+ \times \R^{K+2} \to \R $ by
$$ F_\Psi ( \widetilde y, v, p) = - \sup \big \{-(\rho +\lambda( \widetilde{y}, \nu)) v  +  g(\widetilde{y}, \nu)' p + \ell^{\Psi} (\widetilde y, \nu) \colon \nu \in [0, \numax] \big \}\,.$$
The dynamic programming equation associated with the control problem~\eqref{eq:def-Tprime}  is
\begin{equation} \label{eq:dyn-prog-v-psi}
F_{\Psi} \big ( \widetilde y, v^\Psi(\widetilde y), \nabla v^\Psi  (\widetilde y) \big) = 0 \text{ for } \widetilde y \in \text{int} \widetilde{\mathcal{Y}}, \quad v^\Psi(\widetilde y) = h(\widetilde y) \text{ for } \widetilde y \in \partial  \widetilde{\mathcal{Y}} \,.
\end{equation}
Moreover, since $V^\prime  = \mathcal{T}^\prime V^\prime$, we expect that $V'$ solves in a suitable sense the equation
\begin{equation} \label{eq:HJB-general}
F_{V'} \big ( \widetilde y, V'(\widetilde y), \nabla V' (\widetilde y) \big) = 0, \text{ for } \widetilde y \in \text{int} \widetilde{\mathcal{Y}}, \quad V'(\widetilde y) = h(\widetilde y) \text{ for } \widetilde y \in \partial  \widetilde{\mathcal{Y}}.
\end{equation}
\begin{remark}\label{rem:HJB}
Notice that, equations \eqref{eq:dyn-prog-v-psi} and \eqref{eq:HJB-general} differ in the sense that in \eqref{eq:dyn-prog-v-psi} the function  $F_\Psi$ enters with $\Psi$ fixed, whereas in \eqref{eq:HJB-general} one works with the function $F_{V'}$. This reflects the fact that control problem \eqref{eq:def-Tprime} associated to equation \eqref{eq:dyn-prog-v-psi} has an exogenously given running cost, while in the optimization problem \eqref{eq:optimality-for-Vprime}, leading to equation \eqref{eq:HJB-general}, function $V'$ is the solution of a fixed point equation, and therefore the running cost is endogenous.
\end{remark}
There are two issues  with equations  \eqref{eq:dyn-prog-v-psi} and \eqref{eq:HJB-general}: $v^\Psi$ and  $V'$ are  typically not $\mathcal{C}^1$ functions, and the value  of these functions  on the non-active part $\partial \widetilde{Y} \setminus \Gamma$ of the boundary is  determined  endogenously.
Following \citet{bib:barles-94} we therefore work with the following notion of viscosity solutions.

\begin{definition}\label{def:viscosity-1}
\begin{itemize}\item[1.] A bounded upper semi-continuous (u.s.c.)  function $v$ on $\widetilde{\mathcal{Y}}$ is a {\em viscosity subsolution} of \eqref{eq:dyn-prog-v-psi}, if for all $\phi \in \mathcal{C}^1(\widetilde{\mathcal{Y}}) $ and all  local maxima ${\widetilde y}_0 \in \widetilde{\mathcal{Y}} $ of $v - \phi$ one has
\begin{equation} \label{eq:subsol-1}
\begin{split}
F_{\Psi} \big ( \widetilde y_0, v(\widetilde y_0), \nabla \phi (\widetilde y_0) \big)  &\le 0 \text{ for }\widetilde y_0 \in \text{int} \widetilde{\mathcal{Y}} , \\
\min \big \{  F_{\Psi} \big ( \widetilde y_0, v(\widetilde y_0), \nabla \phi (\widetilde y_0) \big), v(\widetilde y_0) - h(\widetilde y_0 )\big \} & \le 0  \text{ for } \widetilde y_0 \in \partial \widetilde{\mathcal{Y}} .
\end{split}
\end{equation}
A bounded lower semi-continuous (l.s.c.) function $u$ on $\widetilde{\mathcal{Y}}$ is a {\em viscosity supersolution} of \eqref{eq:dyn-prog-v-psi}, if   for all $\phi \in \mathcal{C}^1(\widetilde{\mathcal{Y}}) $ and all  local minima  ${\widetilde y}_0 \in \widetilde{\mathcal{Y}} $ of $u - \phi$ one has
\begin{equation} \label{eq:supersol-1}
\begin{split}
F_{\Psi} \big ( \widetilde y_0, u(\widetilde y_0), \nabla \phi (\widetilde y_0) \big)  &\ge 0 \text{ for }\widetilde y_0 \in \text{int} \widetilde{\mathcal{Y}} , \\
\max \big \{  F_{\Psi} \big ( \widetilde y_0, u(\widetilde y_0), \nabla \phi (\widetilde y_0) \big), u(\widetilde y_0) - h(\widetilde y_0 )\big \} & \ge 0  \text{ for } \widetilde y_0 \in \partial \widetilde{\mathcal{Y}}.
\end{split}
\end{equation}
A {\em viscosity solution} $v^\Psi$ of  \eqref{eq:dyn-prog-v-psi} is either a continuous function on $\widetilde{\mathcal{Y}}$ that is both a sub and a supersolution of \eqref{eq:dyn-prog-v-psi}, or a bounded function with u.s.c.~and l.s.c.~envelopes that are a sub and a supersolution of \eqref{eq:dyn-prog-v-psi}.
\item[2.] A bounded u.s.c.  function $v$ on $\widetilde{\mathcal{Y}}$ is a {\em viscosity subsolution} of \eqref{eq:HJB-general}, if
the relation \eqref{eq:subsol-1} holds for $F= F_v$. Similarly,
a bounded l.s.c. function $u$ on $\widetilde{\mathcal{Y}}$ is a {\em viscosity supersolution} of \eqref{eq:dyn-prog-v-psi}, if \eqref{eq:supersol-1} holds for $F= F_u$. Finally,
$V'$ is a \emph{viscosity solution} of \eqref{eq:HJB-general},  if it is both a sub and a supersolution of  that equation.
\end{itemize}
\end{definition}
Note that Definition \ref{def:viscosity-1} allows for the case that  $v^\Psi(\widetilde{y}_0 ) \neq h(\widetilde {y}_0)$ for certain  boundary points $\widetilde{y}_0\in \partial \widetilde{\mathcal{Y}}$. In particular, if  $F_{\Psi } \big ( \widetilde y_0, v^\Psi(\widetilde y_0), \nabla v^\Psi (\widetilde y_0) \big)  =0 $ in the viscosity sense, \eqref{eq:subsol-1} and \eqref{eq:supersol-1} hold irrespectively of the value of $h(\widetilde{y}_0)$.

\begin{theorem}\label{thm:viscosity} Suppose that  Assumptions~\ref{ass1} and \ref{ass2} hold.  Then
the value function $V'$ is a  continuous viscosity solution of \eqref{eq:HJB-general} in $\widetilde{\mathcal{Y}}$.
Moreover, a
 comparison principle holds for \eqref{eq:HJB-general}: if $v\ge 0 $ is a  subsolution and $u \ge 0 $ a supersolution of \eqref{eq:HJB-general} such that $v(\widetilde{y})/w$ and $u(\widetilde{y})/w$ are bounded on $\widetilde{\mathcal{Y}}$ and such that $v = u = h$ on the active boundary $\Gamma$ of $\widetilde{\mathcal{Y}}$, then $v  \le u $ on $\operatorname{int}\widetilde{\mathcal{Y}}$. It follows that  $V' $ is the only continuous viscosity solution of \eqref{eq:HJB-general}.
\end{theorem}

\begin{proof} First, by Theorem~\ref{thm:fixed_point}, $V'$ is continuous. Moreover, \citet[Theorem~5.2]{bib:barles-94} implies that $V'$ is a viscosity solution of \eqref{eq:dyn-prog-v-psi} with $\Psi = V^\prime$ and hence of equation \eqref{eq:HJB-general}.

Next we prove the  comparison principle.  In order to establish the inequality $v \le u$ we use an inductive argument based on  the monotonicity of $\mathcal{T}^\prime$ and on a comparison result for \eqref{eq:dyn-prog-v-psi}.
Let $u_0 := u$ and define  $u_1=\mathcal{T}u_0$. It follows from \citet[Theorem 5.2]{bib:barles-94}  that $u_1$ is a viscosity solution of \eqref{eq:dyn-prog-v-psi} with $\Psi = u_0$. Moreover,  $u_1(\widetilde{y})/w$ is bounded on $\widetilde{\mathcal{Y}}$ so that $u_1 = h$ on $\Gamma$. Since $u_0$ is a supersolution of \eqref{eq:HJB-general} it is also a supersolution of \eqref{eq:dyn-prog-v-psi} with $\Psi = u_0$. \citet[Theorem 5.7]{bib:barles-94} gives the inequality  $u_1 \le u_0$ on $\operatorname{int}\widetilde{ \mathcal{Y}}$, since   the functions $u^+$ and $u^-$ defined in that theorem coincide in our case.
Define now inductively $u_{n}=\mathcal T' u_{n-1}$, and suppose that $u_{n}\le u_{n-1}$. Then, using the monotonicity of $\mathcal{T}'$, we have
\[
u_{n+1}=\mathcal T' u_n \le\mathcal T' u_{n-1}=u_n.
\]
This proves that $u_{n +1}\le u_{n} $ for every $n\in \bN$. Moreover, as explained in Section~\ref{subsec:optimality-equation}, the sequence $\{u_n\}_{n \in \bN}$ converges to $V^\prime$ ,  so that $u_n \ge V^\prime$ for all $n$. In the same way  we can construct a sequence of functions  $\{v^n\}$ with $v_0 =v$ such that $v^n \uparrow V'$, and we conclude that  $v \le V^\prime \le u$. The remaining statements are clear.
\end{proof}
\begin{remark}
Note that the results in  \citet{bib:davis-farid-99} do not apply directly  to  our case since their assumptions regarding  the behaviour of the vector field $g$ on the lateral boundary are not satisfied in our model. Moreover, \citet{bib:davis-farid-99}
do not give  a comparison principle  for \eqref{eq:HJB-general}.
\end{remark}

Finally, we write the dynamic programming  equation \eqref{eq:HJB-general} explicitly. To this end, we use the fact that $\lambda( \widetilde{y}, \nu) = \sum_{k=1}^K \pi^k \eta^\P (t,e_k,\nu,\R)$, the definition of $g$,
  and the definition of $l^{V'}$ in \eqref{eq:def-fpsi} to obtain
\begin{align}\label{eq:HJB_2}
0 &=  \frac{\partial V'}{\partial t}(t,w,\pi) + \sup \big \{ H (\nu, t,w,\pi, V', \nabla V') \colon \nu \in [\numin, \numax]\big \}\,, \text{ with }
\end{align}
\begin{align}
H (\nu, t,w,\pi, V', \nabla V') &= -\rho V' + \nu (1-f(\nu)) - \nu \frac{\partial  V'}{\partial w}(t,w,\pi)\nonumber  \\
&+\sum_{k,j=1}^K\frac{\partial  V'}{\partial \pi^k}(t,w,\pi)\pi^j \Big( q^{jk}- \pi^k\int_{\mathbb{R}}  u^k(t,\nu,\pi, z) \eta^{\P}(t,e_j,\nu,\ud z)\Big)\nonumber\\
&+\sum_{j=1}^K \pi^j\int_\R  \Delta V'(t,w,\pi,z) \eta^{\P}(t,e_j,\nu,\ud z) , \label{eq:HJB-2-2}
\end{align}
and  $\Delta V'(t,w,\pi,z):=(1+z) V'\left(t,w, (\pi^{i}(1+u^{i}(t,\nu,\pi,z)))_{i=1,...,K}\right)- V'(t,w,\pi).$
This equation coincides with the standard HJB equation associated with  the controlled  Markov process $(W,\pi)$. The advantage of using viscosity solution theory  is that we are able to give a   mathematical meaning to this equation even if $V'$ is merely continuous. This is relevant in our context. Indeed,  in the next section we present a simple example where $V'$ is not $\mathcal{C}^1$.

\subsection{A counterexample.}\label{subsec:counterexample}
We now give an example within a setup where the value function is a viscosity solution of the dynamic programming equation but not a classical solution.
Precisely, we work in the context of Example \ref{exm1}  with linear permanent price impact and deterministic compensator $\eta^\P$.  For simplicity  we let  $\rho=0$, $s=1$, $h(w)\equiv 0$, $f(\nu) \equiv 0$ (zero terminal liquidation value and no temporary price impact). In this case an exogenous upper bound $\numax$ on the liquidation rate need to be imposed, in order to ensure that the set of controls is compact and that a viscosity solution exists (see Remark \ref{rem:supersolution}). Moreover, we assume that  $c^{\text{up}}<c^{\text{down}}$.
The  function $\overline{\eta}^\P$ from  \eqref{eq:drift} is thus given by
$
\overline{\eta}^\P(\nu):=\theta(c^{\text{up}}-c^{\text{down}}(1+a\nu))
$
and $\overline{\eta}^\P(\nu) < 0$ for $\nu >0$. It follows that $S^{\Nu}$ is a supermartingale for any admissible  $\Nu$, and we conjecture that it is optimal to sell as fast as possible to reduce the loss due to the falling bid price. Denote by $\tau(w):= w/\numax$ the minimal time necessary to liquidate the inventory $w$.
The optimal strategy is thus given by  $\nu^*_t=\nu^{{\max}}\I_{\left[0, \tau(w_0)\wedge T\right]}(t)$. Moreover, for $t< \tau(w_0) \wedge T$ one has   $\overline{\eta}^\P(\nu_t)=\overline{\eta}^\P(\nu^{\max})$ and $\bE (S_t^*) = \exp\big(t  \ \overline{\eta}^\P(\numax)\big)$.
Hence we get that
$$
J(\Nu^*) = \int_0^{\tau(w_0) \wedge T} \hspace{-0.1cm}\nu^{\max} \exp\big( u \, \overline{\eta}^\P(\numax)\big)\ud u\,.
$$
Solving this integral we get  the following candidate for the value function
\begin{equation}\label{eq:V'}
V'(t, w) : =\frac{\numax}{\overline{\eta}^\P(\numax)} \Big \{ \exp\big( \overline{\eta}^\P(\numax) (\tau(w) \wedge (T-t))\big) - 1 \Big \}, \quad (t,w)\in [0,T]\times [0, \numax].
\end{equation}
In order to  verify that $V'$  is in fact the value function   we show that $V'$ is a viscosity solution of the  HJB equation \ref{eq:HJB_2}. In the current setting this equation becomes
\begin{equation}\label{eq:hjbV'}
-\frac{\partial V'}{\partial t}-\sup \Big\{ \nu-\nu \frac{\partial V'}{\partial w}+\overline{\eta}^\P(\nu)V' \colon {\nu\in [0, \nu^{\max}]} \Big\}=0.
\end{equation}
First note that $V'$ satisfies the correct terminal and boundary conditions. Define the set
$$ G:= \{ (t,w) \in [0,T] \times [0,w_0] \colon \tau(w) =  (T-t) \} \,.$$
The function $V'$ is $\mathcal{C}^1$ on $ [0,T] \times [0,w_0] \setminus G$, and it is a classical solution of  \eqref{eq:hjbV'} on this set.
However $V'$ is  not differentiable on $G$  and hence not a classical solution everywhere.

Fix some point $(\overline t, \overline w) \in G$.  In order to show that $V'$ is a viscosity solution of the \eqref{eq:HJB_2} we need to verify the subsolution property in this point. (For the supersolution property there is nothing to show as there is no $\mathcal{C}^1$-function $\phi$ such that $V'-\phi$ has a local minimum in $(\overline{t},\overline{w})$.)
Consider $\phi\in \mathcal{C}^1$ such that $V'-\phi$ has a local maximum in $(\overline{t},\overline{w})$.
By considering the left and right derivatives of the functions $t \mapsto (V'-\phi) (t,\overline w)$ respectively
$w \mapsto (V'-\phi) (\overline t, w)$ we get the following inequalities for the partial derivatives of $\phi$
$$ -\nu^{\max}e^{\overline{\eta}^\P(\nu^{\max})(T-\overline{t})}\le \frac{\partial \phi}{\partial t} (\overline t, \overline w) \le 0 \; \text{ and }\;
  0 \le \frac{\partial \phi}{\partial w} (\overline t, \overline w)  \le \exp\big ( \overline{\eta}^\P(\numax) \tau(\overline{w}) \big) \,.
$$
Moreover, it holds  on $G$  that $V'(t,w )=\frac{\nu^{\max}}{\overline{\eta}^\P(\numax)} \big\{ \exp \big( \overline{\eta}^\P(\numax)(T-t)\big) -1\big\}$.  As $w = \numax (T-t)$ on  $G$, differentiating with respect to $t$ gives that
\begin{equation} \label{eq:dphi}
\Big( \frac{\partial \phi}{\partial t}-\nu^{\max}\frac{\partial \phi}{\partial w} \Big )(\overline{t}, \overline{w}) =-\numax \exp \big (\overline{\eta}^\P(\numax)(T-\overline{t})\big ) .
\end{equation}
Applying the inequalities for $\frac{\partial \phi}{\partial w}$ we get  that
\begin{align*}
\sup \Big\{    \nu-\nu\frac{\partial \phi}{\partial w}+\overline{\eta}^\P(\nu)V' \colon {\nu\in [0,\nu^{\max}]}\Big\}
&=\nu^{\max}\Big(-\frac{\partial \phi}{\partial w}+ e^{\overline{\eta}^\P(\nu^{\max})(T-t)}\Big).
\end{align*}
Using \eqref{eq:dphi}	this gives
$-\frac{\partial \phi}{\partial t}-\sup\big \{\nu-\nu\frac{\partial \phi}{\partial w}+\overline{\eta}^\P(\nu)V' \colon \nu\in [0,\nu^{\max}]\big\} =0$
and hence the subsolution property.

\begin{remark} \label{rem:supersolution}It is easily seen that  for $\nu^{\max}\to \infty$ the value function $V'$ from \eqref{eq:V'} converges to $V^{\prime, \infty} (t,w) :=- \frac{1} {\theta\cdown a} \big ( \exp(-w \theta a \cdown)-1 \big)$ and that $V^{\prime, \infty} $ is a strict (classical) supersolution of equation \eqref{eq:hjbV'}  since
\[
\nu-\nu \frac{\partial V'}{\partial w}+\overline{\eta}^\P(\nu)V'=\left(1-e^{-w\theta a \cdown}\right)\left(\frac{\cupp-\cdown}{a\cdown}\right)<0.
\]
Hence we get that  $V^{\prime, \infty} $ is larger or equal than the value function  of the optimal liquidation problem for $\numax = \infty$, by the supersolution property and that it is also smaller or equal than that,  since it is the limit of value functions with bounded maximum selling rate. This allows to conclude that $V^{\prime, \infty} $ is indeed the value function  of the optimal liquidation problem for $\numax = \infty$.
\end{remark}

\section{Examples and numerical results}\label{sec:numerics}

In this section we study  the optimal liquidation rate and the expected liquidation profit in our model. For concreteness  we work in the framework of  Example~\ref{exm2}, that is the example where $\eta^\P$ depends on the liquidation strategy as well as on a two-state Markov chain. We focus on two different research questions:    i.) the influence of model  parameters on the form of the optimal liquidation rate; ii.) the additional liquidation profit from the use of stochastic filtering and a comparison to classical approaches. Moreover, we report the results of a small calibration study.

\paragraph{Numerical method.}
Since equation in \eqref{eq:HJB-2-2} in the setting of Example~\ref{exm2}, cannot be solved analytically, we resort to numerical methods.
We apply an explicit finite difference scheme to solve the HJB equation and to compute the corresponding liquidation strategy.
First, we turn the HJB equation into an initial value problem via time reversion.
Given a time discretization $0=t_0<\dots<t_k<\dots<t_m=T$ we set $V'_{t_0}=h$, and given $V'_{t_k}$, we
approximate the liquidation strategy as follows:
\begin{align}\label{eq:optprob}
 \nu_{t_k}^\ast (w,\pi):= \operatorname{arg max}_{\nu \in [\numin, \numax]} H (\nu, t_k,w,\pi, V'_{t_k}, \nabla^{\text{disc}} V'_{t_k})   \,,
\end{align}
where $\nabla^{\text{disc}}$ is the gradient operator with derivatives replaced by suitable finite differences. In the sequel  we refer to $\nu_{t_k}^\ast$ from \eqref{eq:optprob} as the candidate optimal liquidation rate.
With this we obtain the next time iterate of the value function,
\begin{align}\label{eq:disc}
 V'_{t_{k+1}}=V'_{t_k}+(t_{k+1}-t_k)\,H (\nu_{t_k}^\ast, t_k,w,\pi, V'_{t_k}, \nabla^{\text{disc}} V'_{t_k})\,.
\end{align}
Since the comparison principle holds, as shown in Theorem \ref{thm:viscosity}, and  the value function is the unique viscosity solution of our HJB equation, we get convergence of the proposed procedure to the value function by similar arguments as in \citet{barles1991,dang2014}; details are presented in Appendix \ref{app:num}. The  motivation for using the candidate optimal strategy is as follows: the value function obtained by the  finite difference approximation \eqref{eq:disc} can be viewed as value function in an approximating control problem where the state process follows a discrete time Markov chain, and the candidate optimal strategy \eqref{eq:optprob} is the optimal strategy in the approximating problem, see for instance Chapter~IX of \citet{bib:fleming-soner-06}. The convergence result for the finite difference approximation of the value function suggests that the candidate optimal strategy is nearly optimal in the original problem. A formal analysis of the optimality properties of the candidate optimal strategy is however beyond the scope of the present paper.

\subsection{Candidate optimal liquidation rate.} \label{subsec:opt-liq-rate}

We start by computing the candidate optimal liquidation rate $\nu_{t_k}^\ast$ for Example~\ref{exm2}, assuming  that the temporary price impact is of the form $f(\nu) = \fconst \nu^\fexp$ for $\fexp >0$.
Since $\pi_t^1+\pi^2_t=1$ for all $t\in[0,T]$, we can eliminate the process $\pi^2$  from the set of state variables. In the sequel we denote by  $\pi_t$ the conditional probability of being in the good state $e_1$ at time $t$ and by $V'(t,w,\pi)$ the value function evaluated at the point $(t,w,(\pi,1-\pi))$.
To compute $\nu_{t_k}^\ast$  we substitute the functions $u^i$ given in \eqref{eq:u1-example} and the dynamics of the process  $(\pi_t)_{0\leq t \leq T} $ from  \eqref{eq:dynamics-pi-Ex2} into the general HJB equation \eqref{eq:HJB_2}.
Denote by  $$\pipost_t = \frac{\pi_t \cdown_1}{\pi_t c^\text{down}_1+(1-\pi_t) \cdown_2}, \quad 0\leq t \leq T,$$  the updated (posterior) probability of state $e_1$ given that a downward jump occurs at $t$.
Moreover, denote the discretized partial derivatives  of $V'$ appearing in \eqref{eq:optprob} by $ \frac{\delta V'}{\delta w}$ and $ \frac{\delta V'}{\delta \pi}$.
Substitution into \eqref{eq:HJB-2-2} leads to
\begin{align}\label{eq:nustar}
\nu^\ast_{t_k} &= \operatorname{arg max}_{\nu \in [\numin, \numax]} \big \{ \nu(1-\fconst \nu^\fexp)  - \nu C(t_k,w,\pi) \big \}, \;\text{ with }  \\
C (t_k,w,\pi) &=   \frac{\delta V'}{\delta w} (t_k,w,\pi)  + \frac{\delta V'}{\delta \pi}(t_k,w,\pi)  \pi(1-\pi) a(c^\text{down}_1 - c^\text{down}_2) \\
 & -  \big \{ (1-\theta) V'(t_k,w,\pipost)-V'(t_k,w,\pi)\big \} (\pi c^\text{down}_1+(1-\pi)c^\text{down}_2)a\,.\label{eq:def-cost}
\end{align}
Maximizing \eqref{eq:nustar} with respect to  $\nu$, we  get that
$\nu^\ast_{t_k} = 0$ if $C (t_k,w,\pi)  > 1$;  for $C (t_k,w,\pi)  \le 1$ one has   $\nu^\ast_{t_k}  = \widetilde{\nu}^\ast \wedge \numax$, where  $\widetilde{\nu}^\ast$ solves the equation
\begin{equation} \label{eq:nu-int-sol}
1- \fconst(\fexp +1) \nu^\fexp = C(t_k,w,\pi)\,.
\end{equation}
In our numerical examples we choose $\numax$ large enough so that the constraint $\nu_t \le \numax$ is never binding.
The characterization \eqref{eq:nu-int-sol} of $\nu^\ast_{t_k}$ is very intuitive:~$1-\fconst(\fexp +1) \nu^\fexp$ gives the marginal liquidation  benefit due to an increase in  $\nu$   and  $C(t_k,w,\pi)$ can be viewed as  marginal cost of an increase in $\nu$ (see below).  For  $C (t_k,w,\pi)  \le 1$,  $\widetilde{\nu}^\ast$ is found by equating marginal benefit and marginal cost; for $C (t_k,w,\pi)  > 1$ the marginal benefit is  smaller than the marginal cost  for all $\nu \ge 0$ and $\nu^*_{t_k} = 0$.

The candidate optimal liquidation rate  $\nu^*_{t_k}$ is thus  determined by the marginal cost  $C(t_k,w, \pi)$, and we now   give an economic interpretation of the terms in \eqref{eq:def-cost}.
First, $\frac{\delta V'}{\delta w} $ is a  marginal opportunity cost, since  selling inventory  reduces the amount that can be liquidated in the future. Moreover,  it holds that
\begin{align}\nonumber
- \big ( (1-\theta) V'(t_k,w,\pipost)-V'(t_k,w,\pi)\big ) &= \theta V'(t_k,w,\pipost) - \big( V'(t_k,w, \pipost) - V'(t_k,w, \pi)\big)\,.\label{eq:estimate-liq-cost}
\end{align}
The term  $\theta V'(t_k,w,\pipost)$ represents the reduction in the expected liquidation value due to a downward jump in the return process, and
$(\pi c^\text{down}_1+(1-\pi)c^\text{down}_2)a$ is the marginal increase in the intensity of a downward jump, so that the term
\begin{equation} \label{eq:illiquidity-cost}
\theta V(t_k,w,\pipost)  (\pi c^\text{down}_1+(1-\pi)c^\text{down}_2)a
\end{equation}
measures the marginal  cost due to   permanent price impact; in the sequel we refer to \eqref{eq:illiquidity-cost} as \emph{illiquidity cost}.
Finally, note that
$\pipost - \pi = \frac{ \pi(1-\pi) (c^\text{down}_1 - c^\text{down}_2)}{\pi c^\text{down}_1+(1-\pi)c^\text{down}_2} \,.$
Hence the remaining terms in \eqref{eq:def-cost} are equal to
\begin{equation} \label{eq:uncertainty-cost}
- \Big ( V'(t_k,w, \pipost) - V'(t_k,w, \pi) -  \frac{\delta V'}{\delta \pi}(t_k,w,\pi)(\pipost - \pi ) \Big ) a(\pi \cdown_1+(1-\pi)\cdown_2)\,.
\end{equation}
Simulations indicate that $V'$ is convex in $\pi$; this is quite natural as it implies that uncertainty about the true state reduces the optimal liquidation value. It follows that  \eqref{eq:uncertainty-cost} is negative which  leads to an increase in  the candidate optimal liquidation rate \eqref{eq:nu-int-sol}.  Since $\pipost - \pi $ is largest for $\pi \approx 0.5$, this effect is most pronounced if the investor is uncertain about the true state.   Hence  \eqref{eq:uncertainty-cost} can be viewed as an uncertainty correction that makes the trader sell faster if he is uncertain about the true state.

\paragraph{Numerical analysis and  varying price impact parameters.} To gain further insight on the structure of the candidate  optimal liquidation rate  we resort to numerical experiments.  We work with the parameter set given in Table \ref{tab:params}.
Moreover, we set the liquidation value $h(w)\equiv 0$; this amounts to a strong penalization of any remaining inventory at $T$.
Without loss of generality we set $s=1$, so that the expected liquidation profit is equal to $V'$.

\begin{table}[htbp]
{\small
\begin{center}
\renewcommand\arraystretch{1.5}
 \begin{tabular}{c c c c c c c c c c c}
 \hline
 $w_0$ &  $T$ & $\rho$ & $\theta$ &$\cupp_1, \cdown_2$& $\cdown_1, \cupp_2$ & $a$ &  $\fexp$ & $q^{12}$ & $q^{21}$\\ \hline
  6000  & 2 days & $0{.}00005$ & $0{.}001$ &1000&900& $7\times 10^{-6}$ &  $0.6$ & $4$ & $4$\\
\hline
\end{tabular}
\caption{Parameter values used in numerical experiments.\label{tab:params}}
\label{parameter_values}
\end{center}}
\end{table}

First, we discuss the form of the  candidate optimal liquidation rate for varying size of the temporary price impact, that is for varying $\fconst$, keeping  the permanent price impact parameter $a$ constant at the moderate value $a= 7\times 10^{-6}$.
Figure \ref{fig:pi-cf} shows the liquidation rate at $t=0$ for intermediate and large temporary price impact as a function of $w$ and $\pi$.  The figure is a contour plot:   white areas correspond to $\nu_0 =0$, grey areas correspond to selling at a moderate speed, see also the color bars below the graphs.
Comparing the graphs we see that for higher temporary price  impact (high $\fconst$) the trader tends to trade more evenly over the state space to keep the cost due to the temporary price impact small. The candidate optimal strategy is then characterized by two regions: a \emph{sell} region, where the trader sells at some (varying) speed, and a \emph{wait} region, where she does not sell at all. This reaction of $\nu^*_{t_k}$ to variations in $\fconst$ can also be derived theoretically by inspection of \eqref{eq:nu-int-sol}.

\begin{figure}[ht]
\begin{center}
\includegraphics[width=0.45\textwidth]{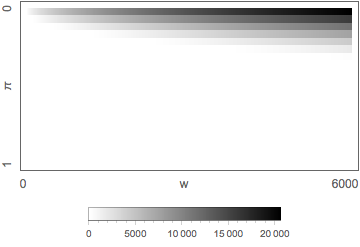}
\includegraphics[width=0.45\textwidth]{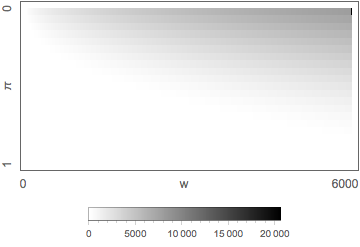}
\caption{Contour plot of the liquidation policy as a function of $w$ (abscissa) and $\pi$ (ordinate) for $\fconst=10^{-5}$ (left), and $\fconst=5\times 10^{-5}$ (right) and $t=0$ for Example~\ref{exm2}.}\label{fig:pi-cf}
\end{center}
\end{figure}

Now we study the impact of the permanent price impact $a$ on the form of the  candidate optimal liquidation rate. Figure~\ref{fig:pi-cf} shows that for  moderate $a$  the liquidation rate is decreasing in $\pi$  and increasing in the inventory level. The situation changes when  the permanent price impact becomes large. Figure~\ref{fig:pi-a} depicts the sell and wait regions under partial information in dependence of the inventory level $w$ and the filter probability $\pi$ for $a=7 \times 10^{-5}$.
For this value of  $a$ the sell region forms a band from low values of $w$ and $\pi$ to high values of $w$ and $\pi$. In particular, for large $w$ and small $\pi$  there is a {\em gambling region} where the trader does not sell, even if a small value of $\pi$ means that the bid price is trending downward (recall that $\pi$ gives the probability that $Y$ is in the good state).

\begin{figure}[ht]
\begin{center}
\includegraphics[width=0.45\textwidth]{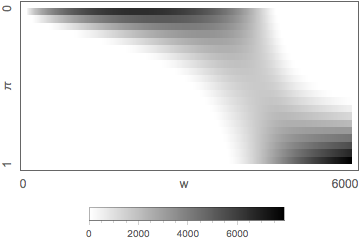}
\caption{Contour plot of the liquidation policy as a function of $w$ (abscissa) and $\pi$ (ordinate) for $\fconst= 10^{-5}$  and for $a= 7 \times 10^{-5}$ and $t=0$ for Example~\ref{exm2}.}\label{fig:pi-a}
\end{center}
\end{figure}

The observed form of $\nu^*_{t_k}$ has the following explanation.  Our numerical experiments show that  for the chosen parameter values  $V'$  is almost  linear in $\pi$, so that the  uncertainty correction \eqref{eq:uncertainty-cost} is negligible. Hence  the liquidation rate $\nu^\ast_{t_k}$ is determined  by the  interplay of the
opportunity cost  $\frac{\delta V'}{\delta w} (t_k,w,\pi)$ and of the illiquidity cost \eqref{eq:illiquidity-cost}.
We found that the opportunity cost is  increasing in $\pi$. This is very intuitive:  in the good state  the investor expects an increase in the expected bid price which makes additional inventory more valuable. Moreover, we found that $\frac{\delta V'}{\delta w} (t_k,w,\pi)$ is decreasing  in $w$, that is the optimal liquidation problem has decreasing returns to scale.
The illiquidity cost  has  the opposite monotonicity behaviour:  it is increasing in $w$ (as it is proportional to $V'(t_k,w,\pipost) $)   and,    for the given parameters,  decreasing in $\pi$.
Now for small values of $a$ the opportunity cost dominates the illiquidity cost  for all $(w,\pi)$ and  $C(t_k,w,\pi)$ is increasing in $\pi$ and decreasing in $w$.   By \eqref{eq:nu-int-sol}, the liquidation rate is thus decreasing in $\pi$ and increasing in $w$, which is in line with the monotonicity behaviour observed  in Figure~\ref{fig:pi-cf}.
If $a$ is large the situation is more involved. The opportunity cost dominates  for  small  $w$, leading to a liquidation rate that is decreasing in $\pi$.  For large $w$ the illiquidity cost dominates,   $C$ is decreasing in $\pi$, and the  candidate optimal liquidation rate is increasing in $\pi$. For $w$ large enough this effect is strong enough to  generate the unexpected gambling region observed in Figure~\ref{fig:pi-a}.

\paragraph{Impact of other model components.} In reality the support of $\eta^\P$   is larger than  $\{- \theta, \theta\}$ as the price may jump by more than one tick. Hence it is important to test the sensitivity of $\nu_{t_k}^\ast$ with respect to the precise form of the  support. To this end,  we computed the candidate optimal strategy for a different parameter set $\tilde \theta, {\tilde c}^{\text{up}}_i , {\tilde c}^{\text{down}}_i$,  $i=1,2$    with $\tilde \theta = 2 \theta $ and ${\tilde c}^{\text{up}}_i = 0.5 \cupp_i$,   ${\tilde c}^{\text{down}}_i = 0.5 \cdown_i$, $i=1,2$.  Note that for the new parameters the  support of $ \eta^\P $ is different but the expected return of the bid price in each of the two states is the same. We found that the liquidation value and the candidate optimal strategy were nearly identical to the original case. This shows that our approach is quite robust with respect to  the exact form of the support of $\eta^\P$ and  justifies the use of a simple model with only two possible values for the jump size  of $R$.

\subsection{Gain from filtering and comparison to classical approaches.} \label{subsec:gain-from-filtering}

In this section we compare the  expected  proceeds  of using the optimal liquidation rate  to the expected proceeds of a trader who mistakenly uses a model with deterministic $\eta^\P$ as in Example~\ref{exm1}. We use the following parameters for the deterministic model:
$c^{\text{up}} = 0.5 c_1^{\text{up}} + 0.5 c_2^{\text{up}}$, $c^{\text{down}} = 0.5 c_1^{\text{down}} + 0.5 c_2^{\text{down}}$, that is the  trader ignores regime switching  but works with the stationary distribution of the Markov chain  throughout, and we set $\fconst=5\times 10^{-5}$ (high temporary price impact). To compute the resulting liquidation rate $\nu^{\ast,\text{det}}_{t_k}$, we consider  the value function $\Vdet$ for  Example~\ref{exm1}. $\Vdet$ is a function  of $t$ and $w$  and it is the unique viscosity solution of the HJB equation
\begin{equation}\label{eq:HJB-for-ex1}
\frac{\partial \Vdet}{\partial t} - \rho \Vdet   + \sup_{\nu \in [0,\numax]}\Big  \{ \nu(1-\fconst \nu^\fexp) - \nu \frac{\partial \Vdet}{\partial w} - \bar\eta^\P (\nu)  \Vdet \Big  \} =0,
\end{equation}
with  $\bar \eta^\P (\nu) = \theta \cdown a $.    Then   $\nu^{\ast,\text{det}}_{t_k}$ is the maximizer in \eqref{eq:HJB-for-ex1} (with partial derivatives replaced by finite differences) and depends only on time and  inventory level.
In our numerical experiments  the expected gain from the use of filtering was equal to
$7.56$\% of the original $w_0=6000$.
This shows that the additional complexity of using a filtering model may be worthwhile.

\begin{remark}[Comparison to \protect{\citet{almgren2001optimal}}] It is interesting that the optimal liquidation rate $\nu^{\ast,\text{det}}_{t_k}$ is identical to the optimal rate in a geometric version of  the well-known model of \citet{almgren2001optimal}, referred to as \emph{geometric AC-model} in the sequel, see, e.g. \citet{gatheral2011optimal} and references therein. In particular,  the performance comparison applies also to the case where the investor uses this classical model.  In the geometric AC-model it is assumed that  that the bid  price has dynamics
\begin{equation} \label{eq:geometric-AC}
dS_t^{\Nu} = \bar \eta^\P (\nu_t) S_t^{\Nu} dt + \sigma S_t^{\Nu} d B_t\,,
\end{equation}
for a Brownian motion $B$. By standard arguments  the HJB equation for the value function $V^\text{AC} $ of the optimal liquidation problem in the geometric AC-model is
\begin{equation} \label{eq:HJB-AC}
\frac{\partial V^\text{AC}}{\partial t}  - \rho V^\text{AC} + \sup_{\nu \in [0,\numax]}\Big  \{ s\nu(1-\fconst \nu^\fexp) - \nu \frac{\partial V^\text{AC}}{\partial w}  - \bar\eta^\P(\nu)  s \frac{\partial V^\text{AC}}{\partial s} + \frac{1}{2} \sigma^2 s^2 \frac{\partial^2  V^\text{AC}}{\partial s^2}  \Big  \} =0\,.
\end{equation}
Moreover, since $V^\text{AC}$ is homogeneous in $s$, $V^\text{AC}(t,s,w) = s V^{\prime, \text{AC}}(t,w)$.  It follows that $\frac{\partial^2  V^\text{AC}}{\partial s^2} =0$, and the HJB equation for $V^{\prime, \text{AC}}$ reduces to~\eqref{eq:HJB-for-ex1}.   Hence the optimal liquidation rate in the geometric AC model and in the jump-model with deterministic compensator coincide.
Note that the equivalence between the jump model and the  geometric AC-model \eqref{eq:geometric-AC} holds only for the case where the compensator is deterministic: a model of the form  \eqref{eq:geometric-AC} with drift driven by an unobservable Markov chain would lead to a diffusion equation for the filter and hence to a control problem for diffusion processes.
\end{remark}

\subsection{Model calibration.} \label{subsec:calibration}
Finally we report the results of a small calibration study. We used a robust version of the EM algorithm to estimate the parameters  of  the bid price dynamics for the model specification from Example~\ref{exm2}; see \citet{bib:damian-eksi-frey-17} for details on the methodology. First, in order to test the performance of the algorithm we ran a study with simulated data for  two different parameter sets. In set~1 we use the parameters from Table~\ref{tab:params}; in set~2 we work with  $\cupp_1 =\cupp_2=\cdown_1=\cdown_2 =1000$, that is we consider a situation without Markov switching in the  true data-generating process. However, the  EM algorithm  allows for different parameters in the two states, so that parameter set~2  is a  test, if the EM  methodology points out  spurious regime changes and trading opportunities which are not really in the data. The outcome of this exercise is presented in Figure~\ref{fig:filters}, where we plot the hidden trajectory of $Y$ together with the filter estimate $\widehat{Y}$ generated from the simulated data using  the estimated model parameters.  We see that in the left plot the filter nicely picks up the regime change, in the right plot the estimate $\widehat{Y}_t$  is close to $1.5$ throughout, that is the estimated model correctly indicates  that there is no Markov switching in the data.
Finally we applied the algorithm  to bid price data from the share price of Google, sampled at a  frequency of one second. The EM  estimates are $\widehat{c}^\text{up}_1 = 2128$, $\widehat{c}^\text{up}_2 = 1751$, $\widehat{c}^\text{down}_1 = 1769$,  $\widehat{c}^\text{down}_2 = 1888$, which shows the same qualitative behaviour as the values used in our simulation study. A trajectory of the ensuing filter is given in Fig~\ref{fig:filters-google}.

One would need an extensive empirical study to confirm and refine these  results, but this is beyond the scope of the present paper.

\begin{figure}[ht]
\begin{center}
\includegraphics[scale=0.22]{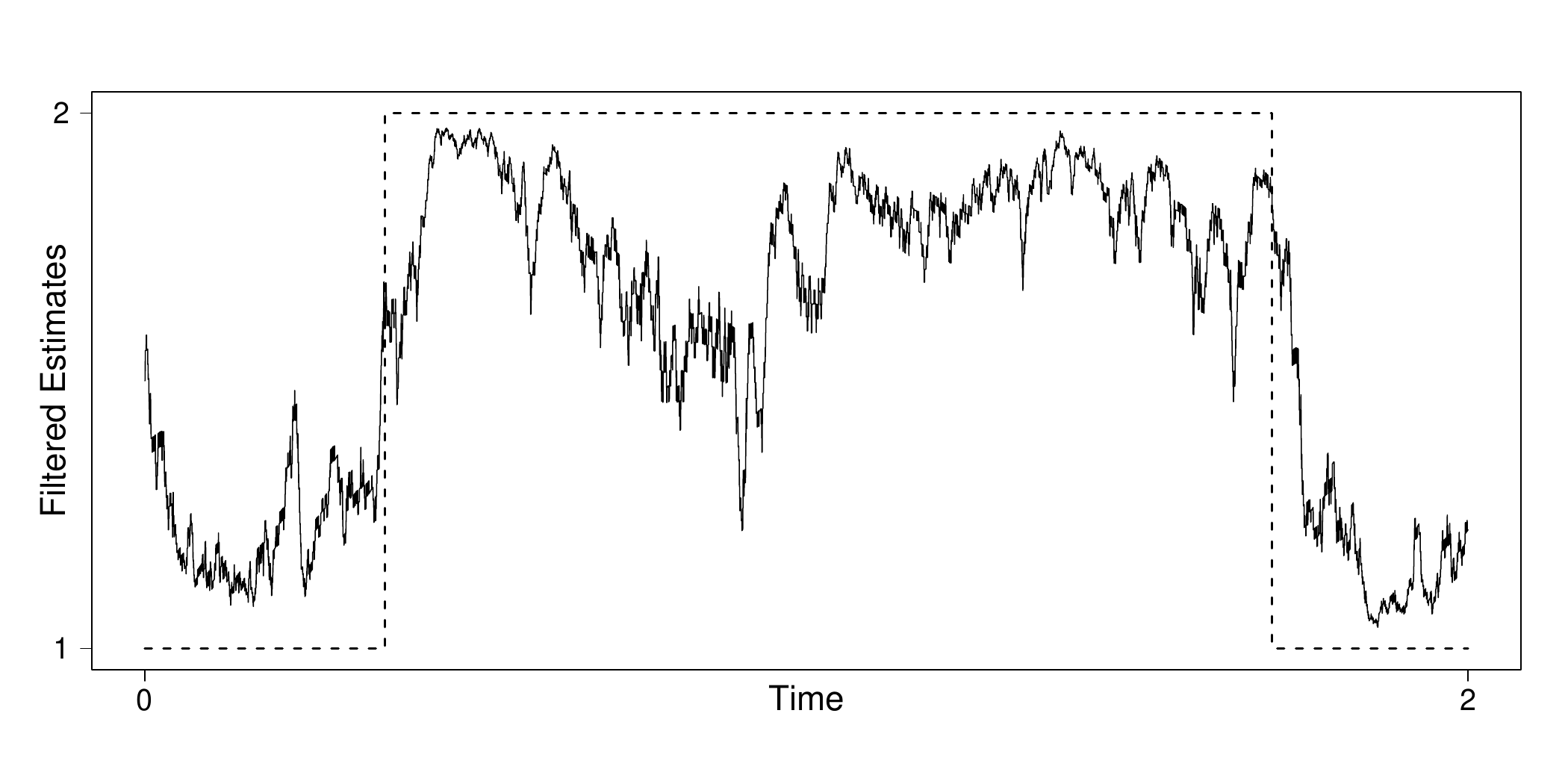}
\includegraphics[scale=0.22]{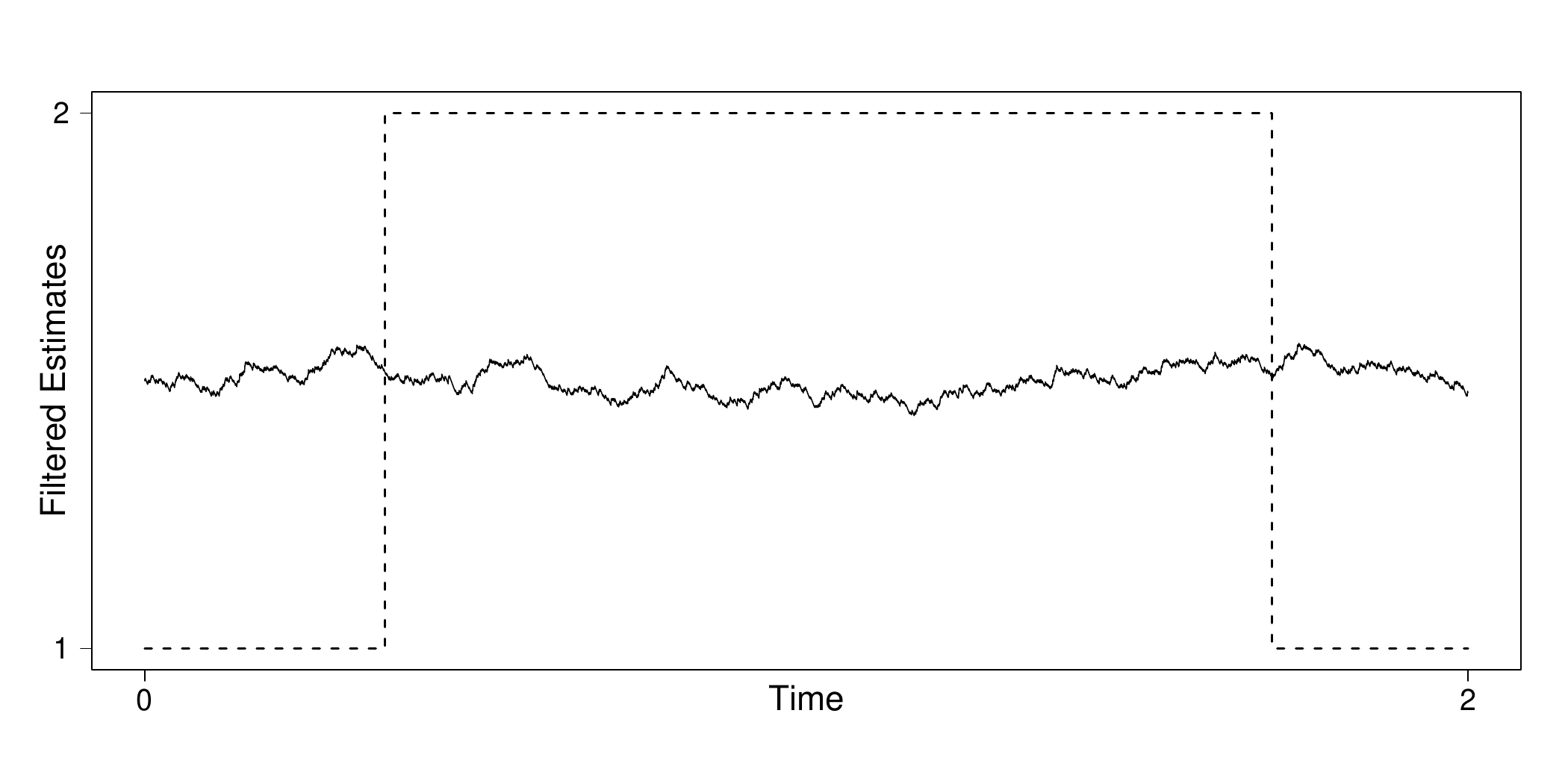}
\caption{A trajectory of the Markov chain $Y$ (dashed)  and of the corresponding filter $\widehat{Y}$ (straight line)  computed using the parameter estimates from the EM  algorithm as input. Left plot: results for  parameter set~1 (with Markov switching); right plot: results for parameter set~2 (no Markov switching) In the graphs state $e_1$ ($e_2$)  is represented by the value 1  (the value 2), and $\widehat{Y}_t = \pi_t 1 + (1-\pi_t)2$.
The estimated parameters for parameter set ~1 are as follows: $\widehat{c}^\text{up}_1 = 993$;  $\widehat{c}^\text{up}_2 = 875$; $\widehat{c}^\text{down}_1 = 842$; $\widehat{c}^\text{down}_2 = 960$. For parameter set~2 we obtained $\widehat{c}^\text{up}_1 = 940$;  $\widehat{c}^\text{up}_2 = 941$; $\widehat{c}^\text{down}_1 = 945$; $\widehat{c}^\text{down}_2 = 957$. \label{fig:filters}
 }
 \end{center}
 \end{figure}

 \begin{figure}[h]
 \begin{center}
\includegraphics[scale=0.22]{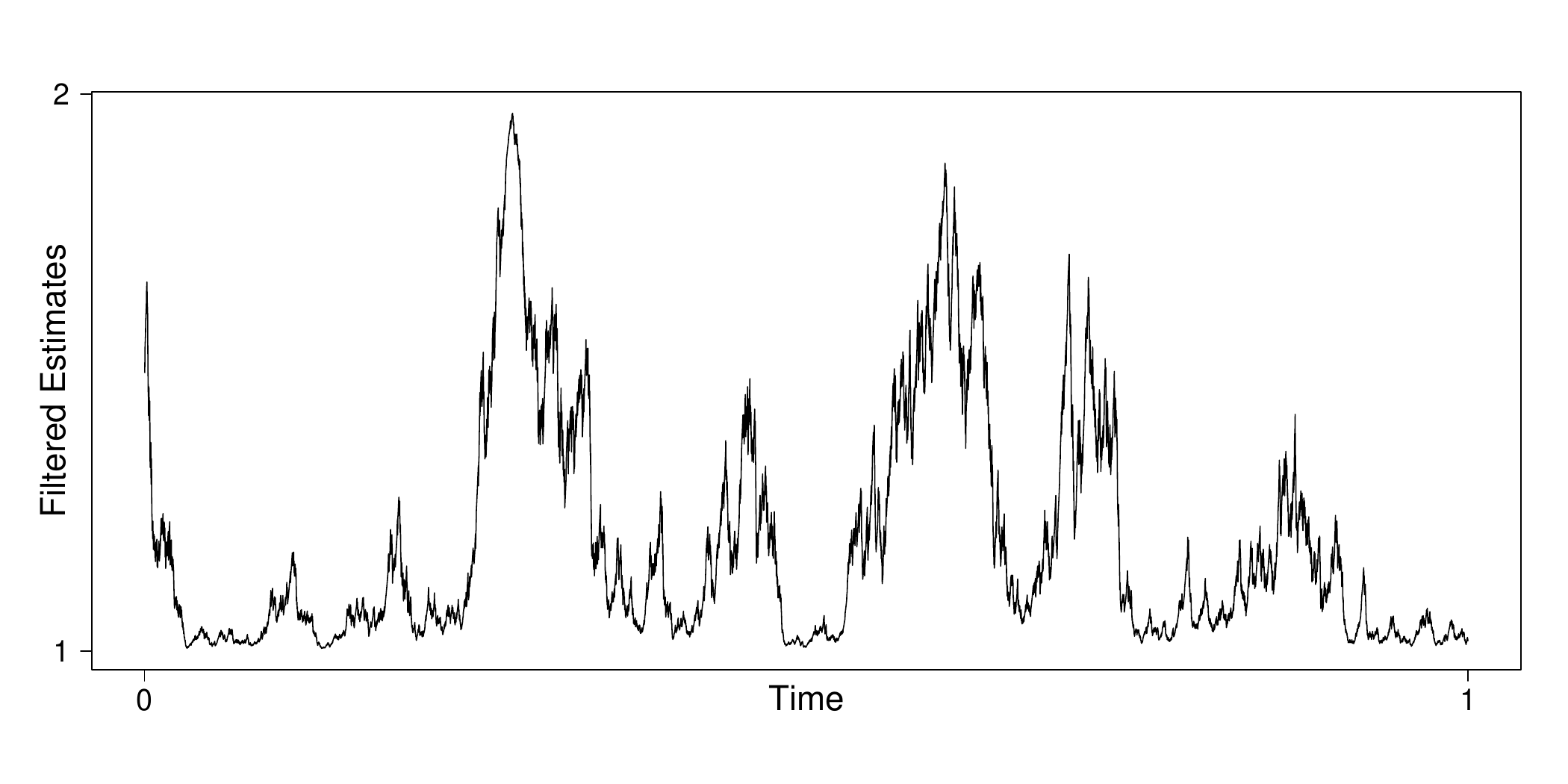}
\caption{Trajectory  of $\widehat{Y}$  computed from  the Google share price on  2012-06-21,  sampled at a frequency of one second. (Data are from the LOBSTER database, see {\tt https://lobsterdata.com}) }
 \label{fig:filters-google}
\end{center}
\end{figure}

\section*{Acknowledgements}
The authors are grateful for useful suggestions from several anonymous referees and for the excellent research assistance by Camilla Damian.
Support by the Vienna Science and Technology Fund (WWTF) through project MA14-031 is gratefully acknowledged.
The work of K.~Colaneri was partially supported by INdAM GNAMPA through projects UFMBAZ-2017/0000327 and UFMBAZ-2018/000349.
A part of this article was written while K.~Colaneri was affiliated with the Department of Economics, University of Perugia, Via A. Pascoli 20, 06123 Perugia, Italy and with the School of Mathematics, University of Leeds, LS2 9JT, Leeds, UK.
M.~Sz\"olgyenyi is supported by the AXA Research Fund grant ``Numerical Methods for Stochastic Differential Equations with Irregular Coefficients with Applications in Risk Theory and Mathematical Finance".
A part of this article was written while M.~Sz\"olgyenyi was affiliated with the Seminar for Applied Mathematics and the RiskLab Switzerland, ETH Zurich, R\"amistrasse 101, 8092 Zurich, Switzerland and with the Institute of Statistics and Mathematics, Vienna University of Economics and Business, Welthandelsplatz 1, 1020 Vienna, Austria.

\appendix

\section{Setup and filtering: proofs and additional results}\label{app:filtering}

\begin{lemma}\label{lemma:martingale-part-of-S}   Suppose that Assumption~\ref{ass1} holds.  Fix $m > w_0/T$ and consider some $\bF^S$-adapted strategy $\Nu$ with values in  $[0,m]$. Define
$$C:= 0 \vee \sup \big \{\int_\R (z^2 + 2z) \eta^\P (t,e,\nu,dz)\colon  (t, e, \nu)  \in [0,T] \times \mathcal{E} \times  [0 ,  m] \big \}\,. $$  Then   $C < \infty$,  $\mathbb{E}((S_t^{\Nu})^2) \le S_0^2 e^{C t}$,  and $(\int_0^t S_{s-}^{\Nu} \ud M_s^R)_{0\le t\le T}$ is a true martingale.
\end{lemma}

\begin{proof}  To ease the notation we write $S_t$ for $S_t^{\Nu}$.    We begin with the bound on $S_t^2$. First note that $C$ is finite by Assumption~\ref{ass1}.  At a jump time $T_n$ of $R$ it holds that $S_{T_n} = S_{T_{n}-}(1 + \Delta R_{T_n})$ and therefore
$$S_{T_n}^2 - S_{T_n -}^2 = S_{T_n -}^2 \Delta R_{T_n} ^2 + 2 S_{T_n -}^2 \Delta R_{T_n}\,.$$
Hence $ S_t^2 = S_0^2 + \int_0^t \int_\R S_{s-}^2 (z^2 + 2z) \mu^R(\ud z,\ud s)$ and we get
\begin{align}
 \mathbb{E}(S_t^2) & = S_0^2 + \mathbb{E} \Big ( \int_0^t  \int_\R S_{s}^2 (z^2 +2z) \eta^\P (s, Y_{s-} \nu_{s-}, \ud z) \, \ud s \Big )
       \\&  \le S_0^2 + C \int_0^t \mathbb{E} \big( S_s^2\big) \ud s\,,
\end{align}
so that $\mathbb{E}((S_t^{\Nu})^2) \le S_0^2 e^{C t}$ by the Gronwall inequality. To show that
$\int_0^{\cdot} S_{s-} \ud M_s^R$ is a true martingale we show that this process has integrable quadratic variation. Since
 $ \big[ \int_0^\cdot   S_{s-} \ud M_s^R \big]_t = \int_0^t \int_\R S_{s-}^2 z^2 \mu^R(\ud z, \ud s) $, we have
$$ \mathbb{ E} \left ( \big[ \int_0^\cdot   S_{s-} \ud M_s^R \big]_t \right ) =
\mathbb{E} \left ( \int_0^t S_{s}^2 \int_\R z^2 \eta^\P(s, Y_{s-}, \nu_{s-}\ud z) \ud s\right) \le S_0^2 \tilde C \int_0^t e^{C s} \ud s\,,$$
for every $t\in[0,T]$,
where ${\tilde C} = \sup \big \{\int_\R z^2  \eta^\P (t,e,\nu,dz)\colon  t \in[0,T], \ e \in \mathcal{E}, \nu \in [0 ,  m] \big \}$ is finite by Assumption~\ref{ass1}.
\end{proof}

\begin{proof}[Proof of Lemma \ref{lem:change-of-measure}] Conditions~\eqref{eq:cond-on-density} and \eqref{eq:bound_etaP} imply that $\tilde Z$ is a true martingale, see \citet{protter2008arbitrage}.  Moreover, $\beta( t,Y_{t^-},\nu_{t^-}, z) > - 1$, since $\big( \ud \eta^\P_t (t,e_i, \nu;\ud z)/\ud \eta^\Q_t(\ud z) \big)(z) >0$ by assumption. This implies that $\widetilde{Z}_T >0$, and hence the equivalence of $\P$ and $\Q$.
The Girsanov theorem for random measures (see \cite[VIII, Theorem T10]{bremaud1981point})  shows that under $\P$, $\mu^{R} (\ud t,\ud z)$ has the predictable compensator
$(\beta(t, Y_{t-}, \nu_t, z) +1) {\eta}_t^\Q (\ud z) \ud t.$ By definition of $\beta$  this is equal to ${\eta}^P(t, Y_{t-}, \nu_t, \ud z) \ud t \,.$ Moreover, $\widetilde Z$ and $Y$ are orthogonal, since $R$ and $Y$ have no common jumps, so that the law of $Y$ is the same under $\P$ and under $\Q$.

\end{proof}

%
%

\begin{proof}[Proof of Theorem \ref{thm:Zakai}] Our derivation parallels the proof of \cite[Theorem~3.24]{bain2009filtering}, which deals with  the classical  case where the observation process is a Brownian motion with drift.
Recall that for a function $f \colon \mathcal{E} \to \R$  the semimartingale decomposition of $f(Y_t)$ is given by
$f(Y_t)=f(Y_0)+\int_0^t \langle Q\mathbf{f}, Y_s \rangle\ud s+M_t^f$,
where $M^f$ is a true $(\bF, \Q)$-martingale.
Define the process $\widetilde Z^\epsilon=(\widetilde Z^\epsilon_t)_{0\leq t \leq T}$ by
\[\widetilde Z^\epsilon_t:=\frac{\widetilde Z_t}{1+\epsilon\widetilde Z_t},\]
and note that $\widetilde Z^\epsilon_t < 1/\epsilon$ for every $t \in [0,T]$.
Now we compute $\widetilde Z^\epsilon f(Y)$. Notice that $[\widetilde Z^\epsilon , Y]_t =0$ for every $t \in [0,T]$, as $R$ and $Y$ have no common jumps. Hence, from  It\^o's product rule we get
\begin{align}
\ud  \big(\widetilde Z^\epsilon_t f(Y_t)\big) &=
 \widetilde Z^\epsilon_{t^-} \langle Q\mathbf{f}, Y_t \rangle \ud t + \widetilde Z^\epsilon_{t^-} \ud M^f_t
  - f(Y_{t^-}) \widetilde Z^\epsilon_{t^-} \int_\R\frac{\beta(t,Y_{t^-},\nu_{t-},z)}{1 + \epsilon  \widetilde Z_{t^-}} \eta^\Q_t(\ud z) \ud t \\
& \quad + f(Y_{t^-}) \widetilde Z^\epsilon_{t^-}\int_\R \frac{\beta(t,Y_{t^-},\nu,z)}{1+\epsilon \widetilde Z_{t^-} (1+\beta(t,Y_{t^-},\nu_{t-},z))} \mu^R(\ud t, \ud z).\label{eq:zakaihelp}
\end{align}

Next we show that $\qesp{\int_0^t\widetilde Z^\epsilon_{s^-}\ud M^f_s\mid\F^S_t}=0$. By the definition of conditional expectation, this is equivalent to
$\qesp{H\int_0^t\widetilde Z^\epsilon_{s^-}\ud M^f_s}=0 $ for every bounded, $\F_t^S$-measurable random variable $H$. Define an $(\bF^S,\Q)$-martingale by $H_u=\qesp{H\mid\F^S_u}$, $0\le u \le t\le T$, and note that $H=H_t$. By the martingale representation theorem for random measures, see, e.g., \cite[Ch.
III, Theorem 4.37]{jacod2003limit} or \cite[Ch.~VIII, Theorem T8]{bremaud1981point}, we get that  there is a bounded $\bF^S$-predictable random function  $\phi$ such that  \[H_t=H_0+\int_0^t\int_{\R}\phi(s,z)(\mu^R(\ud s, \ud z)-\eta^{\Q}(\ud z)\ud s), \quad 0\leq t \leq T \,. \]
Now, applying the It\^{o} product rule and using that $[M^f, H]_t = [Y,R]_t =0$ for every $t \in [0,T]$,  we obtain
\[
H_t \int_0^t\widetilde Z^\epsilon_{s^-}\ud M^f_s\!=\!\int_0^t H_{s-}\widetilde Z^\epsilon_{s^-}\ud M^f_s+\int_0^t\!\int_{\R} \!\Big(\!\int_0^s\widetilde Z^\epsilon_{u^-}\ud M^f_u\Big)\phi(s,z)\big(\mu^R(\ud s, \ud z)-\eta^{\Q}(\ud z)\ud s\big).
\]
Both integrals on the right hand side of the above representation are martingales. This follows from  the finite-state property of the Markov chain $Y$ and the boundedness of $\widetilde Z^\epsilon$ and $H$. Hence, taking the expectation we get that $\qesp{H\int_0^t\widetilde Z^\epsilon_{s^-}\ud M^f_s}=0 $ as claimed.

Now note that  for $t \in [0,T]$ and a generic integrable $\F_t$-measurable random variable $U$ it holds that
\begin{equation} \label{eq:cond-esp}
\qesp{U\mid\F^S_t}=\qesp{U\mid\F^S_T};
\end{equation}
this can be shown with similar arguments as in \cite[Proposition 3.15]{bain2009filtering}.
Taking the conditional expectation from \eqref{eq:zakaihelp} and applying \eqref{eq:cond-esp} and Fubini Theorem  we get for every $t \in [0,T]$,
\begin{align}
\nonumber&\qesp{\widetilde Z^\epsilon_t f(Y_t) \mid \F^S_t}  =\frac{\pi_0 (f)}{1+\epsilon}+\int_0^t\qesp{\widetilde Z^\epsilon_{s^-} \langle Q\mathbf{f}, Y_s \rangle \mid\F^S_T}\ud s\\
\nonumber&\, + \int_0^t \int_\R \qesp{ f(Y_{s^-}) \widetilde Z^\epsilon_{s^-} \frac{\beta(s,Y_{s^-},\nu_{s^-},z)}{1+\epsilon \widetilde Z_{s^-} (1+\beta(s,Y_{s^-},\nu,z))}\mid\F^S_T}\mu^R(\ud s, \ud z) \\
&\, -\int_0^t\int_\R\qesp{f(Y_{s^-}) \widetilde Z^\epsilon_{s^-} \frac{\beta(s,Y_{s^-},\nu_{s^-},z)}{1 + \epsilon  \widetilde Z_{s^-}}  \mid\F^S_T}\eta^\Q_s(\ud z)\ud s\,.\label{eq:zakaihelp2}
\end{align}
Note that, for every $t \in [0,T]$, $\widetilde Z^\epsilon_t < \widetilde Z_t $ and that  $\widetilde Z_t $ is integrable. Since $\beta$ is bounded by assumption, by dominated convergence we get the following three limits
\begin{align*}
& \lim_{\epsilon \to 0} \qesp{\widetilde Z^\epsilon_t f(Y_t)\mid\F^S_t}   =\qesp{\widetilde Z_t f(Y_t)\mid\F^S_t}\,, \\
& \lim_{\epsilon \to 0} \int_0^t\qesp{\widetilde Z^\epsilon_{s^-} \langle Q\mathbf{f}, Y_s \rangle\mid\F^S_T}\ud s = \int_0^t\qesp{\widetilde Z_{s^-} \langle Q\mathbf{f}, Y_s \rangle\mid\F^S_T}\ud s\,,\\
&\lim_{\epsilon \to 0} \int_0^t \int_\R\qesp{f(Y_{s^-}) \widetilde Z^\epsilon_{s^-} \frac{\beta(s,Y_{s^-},\nu_{s^-},z)}{1 + \epsilon  \widetilde Z_{s^-}}  \mid\F^S_T} \eta^\Q_s(\ud z)\ud s \\   & \qquad =
\int_\R\qesp{f(Y_{s^-}) \widetilde Z_{s^-} \beta(s,Y_{s^-},\nu_{s^-},z) \mid \F^S_T} \eta^\Q_s(\ud z)\ud s.
\end{align*}

Finally we consider the integral with respect to $\mu^R(\ud s, \ud z)$ in \eqref{eq:zakaihelp2}.  Let $\{T_n, Z_n\}$ be the sequence of jump times and the corresponding jump sizes of the process $R$.
Denote by $n(t) $ the number of jumps up to time $t$, so that $T_{n(t)}$ is the last jump time before $t$. Then
\begin{align*}
&\lim_{\epsilon \to 0}\int_0^t \int_\R \qesp{ f(Y_{s^-}) \widetilde Z^\epsilon_{s^-} \frac{\beta(s,Y_{s^-},\nu,z)}{1+\epsilon \widetilde Z_{s^-} (1+\beta(s,Y_{s^-},\nu,z))}\mid\F^S_T} \mu^R(\ud s, \ud z)\\
&\qquad = \lim_{\epsilon \to 0} \sum_{n=1}^{n(t)} \qesp{ f(Y_{T_n^-}) \widetilde Z^\epsilon_{T_n^-} \frac{\beta(T_n,Y_{T_n^-},\nu,\Delta R_{T_n})}{1+\epsilon \widetilde Z_{T_n^-} (1+\beta(T_n,Y_{T_n^-},\nu_{T_n^-},\Delta R_{T_n}))}\mid\F^S_T} \\
&\qquad  = \sum_{n=1}^{n(t)} \qesp{ f(Y_{T_n^-}) \widetilde Z_{T_n^-} \beta(T_n,Y_{T_n^-},\nu_{T_n^-},\Delta R_{T_n})\mid\F^S_T}\\
&\qquad  = \int_0^t \int_\R \qesp{ f(Y_{s^-}) \widetilde Z_{s^-} \beta(s,Y_{s^-},\nu_{s^-},z)\mid\F^S_T} \mu^R(\ud s, \ud z)\,.
\end{align*}
Assembling the previous results we obtain
\begin{align*}
&\qesp{\widetilde Z_t f(Y_t)\mid\F^S_T}=\pi_0 (f)  +\int_0^t\qesp{\widetilde Z_{s^-} \langle Q\mathbf{f}, Y_s \rangle\mid\F^S_T}\ud s\\
&\quad+\int_0^t\int_\R\qesp{f(Y_{s^-}) \widetilde Z_{s^-} \beta(s,Y_{s^-},\nu_{s^-},z)\mid\F^S_T}\left(\mu^R(\ud s, \ud z)-\eta^\Q_s(\ud z)\right)\,,
\end{align*}
and hence the claim of the theorem follows from \eqref{eq:cond-esp}.
\end{proof}

\section{Optimization via MDMs: proofs and additional results}\label{app:sec-4}

\begin{proof}[Proof of Lemma \ref{lemma:lipschitz}]
To establish the claim we show that the first derivatives of the vector field $g$ are bounded, uniformly in $\nu$.
The components of $\frac{\partial g}{\partial w}$  and $\frac{\partial g}{\partial s}$ are all 0, and, using Assumption \ref{ass1}, the nonzero components of $\frac{\partial g}{\partial \pi^i}$, $i=1,\dots, K$, can be estimated as follows. For $i \neq k$,
\begin{align*}
\left|\frac{\partial g^{k+3}}{\partial \pi^i}\right|\!&=\! \left|q^{ik}\!-\!\pi^k \!\!\int_{\mathbb{R}}\!\!  u^k(t,\nu,\pi,z)\eta^\P(t,e_i,\nu,\ud z)-\pi^k \sum_{j=1}^K \pi^j \!\!\int_{\mathbb{R}} \!\! \frac{\partial u^k(t,\nu,\pi,z)}{\partial \pi^i} \eta^\P(t,e_j,\nu,\ud z)\right|\\
<&\max_{i,k}q^{ik}+\pi^k \int_{\mathbb{R}}  u^k(t,\nu,\pi,z)\eta^\P(t,e_i,\nu,\ud z)\\ &+ \pi^k \sum_{j=1}^K \pi^j \int_{\mathbb{R}} \frac{{\ud \eta^{\P}(t,e_i,\nu)}/{\ud \eta^\Q_t}( z){\ud \eta^{\P}(t,e_k,\nu)}/{\ud \eta^\Q_t}( z)}{\left(\sum_{l=1}^K \pi^l {\ud \eta^{\P}(t,e_l,\nu)}/{\ud \eta^\Q_t}( z)\right)^2}\eta^\P(t,e_j,\nu,\ud z),
\end{align*}
and this is smaller than $\max_{i,k}q^{ik}+(M^4+M^2)\lambda^{\max}$. For $i=k$ we get
\begin{align*}
\left|\frac{\partial g^{i+3}}{\partial \pi^i}\right|=&\left| q^{ii}-2\pi^i \int_{\mathbb{R}}  u^i(t,\nu,\pi,z)\eta^\P(t,e_i,\nu,\ud z)- \sum_{j\neq i}\pi^j\int_{\R}u^i(t,\nu,\pi,z)\eta^\P(t,e_j,\nu,\ud z) \right.\\
&\left.-\pi^i \sum_{j=1}^K \pi^j \int_{\mathbb{R}}  \frac{\partial u^i(t,\nu,\pi,z)}{\partial \pi^i} \eta^\P(t,e_j,\nu,\ud z)\right|<\max_{i}q^{ii}(M^4+3 M^2)\lambda^{\max}\,.
\end{align*}
\end{proof}

\begin{proof}[Proof of Lemma \ref{lemma:bounding-function}] First we estimate the reward function introduced in \eqref{rewardmdm}.  Since $f \ge 0$, $e^{- \rho t} \le 1$,  and  $h(w)\le w$, we get that
$
r(\widetilde{x}, \alpha)  \le s \int_0^{\tau^\varphi} e^{-\Lambda^{\alpha}_u}  \alpha_u \ud u +  s e^{-\Lambda^{\alpha}_{\tau^\varphi}} w_{\tau^\varphi}^\alpha$ .
Partial integration gives
$$ \int_0^{\tau^\varphi} e^{-\Lambda^{\alpha}_u}  \alpha_u \ud u =
\big[- w_u^\alpha e^{-\Lambda^{\alpha}_u}\big]_0^{\tau^\varphi} - \int_0^{\tau^\varphi} \lambda^\alpha_u  e^{-\Lambda^{\alpha}_u} w_u^\alpha \ud u \le w - e^{-\Lambda^{\alpha}_{\tau^\phi}} w_{\tau^\phi}^\alpha\, ,
$$
and hence   $r(\widetilde{x}, \alpha) \le s w$.  Next we  estimate $Q_L b (\widetilde{x},\alpha)$. Recall the definition of $\bar \eta^\P$ from \eqref{eq:drift} and let
$c_{\eta} := \sup \big \{ \bar \eta^\P(t,e,0)   \colon (t,e)  \in [0,T] \times \mathcal{E} \big\}.$
It holds that
\begin{align*}
Q_L b (\widetilde{x},\alpha) &= \int_0^{\tau^{\varphi}} \hspace{- 0.2cm} e^{\gamma(T-(u+t))}e^{-\Lambda^{\alpha}_u}
 \sum_{j=1}^K \pi_j  s w_u^\alpha (1 + \bar \eta^\P( t+u, e_j,  \alpha_u)) \ud u
{\le} s w e^{\gamma(T-t)} c_\eta  \int_0^{\tau_\varphi} \hspace{- 0.2cm}  e^{-\gamma u}\,\ud u\,,
\end{align*}
where we have used  that $w_u^\alpha \le w$ and $e^{-\Lambda^{\alpha}_u}<1$. The last term is bounded by  $ b(\widetilde{x})\frac{c_{\eta}}{\gamma}$, and  the MDM is contracting for $\gamma >c_\eta$.
\end{proof}

The following lemma is needed in the proof of Proposition~\ref{prop:continuity}.

\begin{lemma} \label{lemma:continuity-of-QL}
Consider a  function $v\in \mathcal{C}_b$. Then  the mapping $(\widetilde{x}, \nu) \mapsto \bar Q v(\widetilde x, \nu)$ is continuous on $\widetilde{\mathcal X} \times [0, \numax]$.
\end{lemma}
\begin{proof}
It suffices to show that for $j=1,\dots,K$ the mapping
$$ (t,w,s,\pi,\nu) \mapsto \int_\R v \left (t,s(1+z), \pi^1(1+ u^1(t, \nu, \pi,z), \dots , \pi^K(1+ u^K(t, \nu, \pi,z)\right ) \eta^j(t, \nu, \ud z)
$$
is continuous on $ \widetilde{\mathcal{X}} \times [0, \numax]$, where $\eta^j(t, \nu, \ud z) := \eta(t,e_j,\nu,\ud z)$.
Consider a sequence with elements $ (t_n, \nu_n, \pi_n) \xrightarrow[ n \to \infty]{} (t, \nu, \pi)$. Note that, for sufficiently large $n$,  the set $\{s^n (1+z)\colon z \in \text{supp}(\eta) \}$ is  contained in a compact subset $[\underline{s}, \overline{s}] \subset (0, \infty)$. Moreover, $v$ is uniformly continuous on the compact set
$ [0,T] \times [0, w_0] \times [\underline{s}, \overline{s}] \times \S^K \times [\numin, \numax]\,.
$
Then, Assumption~\ref{ass2}-(2) implies that the sequence $\{v^n\} $ with
$$
v^n(z) :=  v \left (t_n,s_n(1+z), \pi_{n}^1(1+ u^1(t_n, \nu_n, \pi_n,z), \dots \pi^K_n(1+ u^K_n(t_n, \nu_n, \pi_n,z)\right )
$$
converges uniformly in $z \in \text{supp}(\eta)$  to $v(z) := v(t,s,\pi, \nu, z)$. Hence the following estimate holds:
\begin{align}\nonumber
&\Big | \int_{\text{supp}(\eta)} \hspace{-0.2cm} v^n(z)  \eta^j(t_n, \nu_n,\ud z) - \int_{\text{supp}(\eta)} \hspace{-0.2cm} v(z)  \eta^j(t, \nu,\ud z) \Big |  \\
& \label{eq:estimate-of-integral}
 \le \int_{\text{supp}(\eta)} \hspace{-0.7cm} \big |v^n (z) - v(z) \big | \eta^j(t_n, \nu_n , \ud z) +
\Big |\int_{\text{supp}(\eta)} \hspace{-0.7cm} v(z)  \eta^j(t_n, \nu_n,\ud z) -
\int_{\text{supp}(\eta)} \hspace{-0.7cm} v(z)  \eta^j(t, \nu,\ud z) \Big |\,.
\end{align}
Finally,  the   first term in \eqref{eq:estimate-of-integral} can be estimated by
$ \lambda^{\text{max}} \sup\{ |v^n(z) - v(z)| \colon z \in \text{supp}(\eta) \} $, which converges to zero as $v^n$ converges to $v $ uniformly; the second term in \eqref{eq:estimate-of-integral} converges to zero by Assumption~\ref{ass2}-(1) (continuity of the mapping  $(t,\nu) \mapsto \eta^j(t,\nu, \ud z) $ in the weak topology).
\end{proof}

\section{Convergence of the finite difference approximation}
\label{app:num}

\citet{barles1991} introduced conditions under which a numerical scheme converges to the viscosity solution of an HJB equation.
These conditions are {\it consistency}, which means that the difference operators converge to the differential operators, {\it stability},  that is the finite difference operator stays bounded as the time and space steps converge to zero, and {\it monotonicity}, which means that the mapping $V_{t_k}\mapsto V_{t_{k+1}}$ from \eqref{eq:disc} is monotone.
Moreover, a comparison principle for the limiting HJB equation needs to hold.

In the sequel we present the discretization scheme for Example \ref{exm2}, treated numerically in Section \ref{sec:numerics}, and verify the above conditions.
For the discretization of the differential operators we use the standard difference scheme as stated in \citet[Chapter IX]{bib:fleming-soner-06}.
In order to satisfy the monotonicity condition we need to discretize the first order terms appropriately by using an upwinding scheme: depending on the sign of the coefficient, we use a forward or a backward difference operator.
In Example \ref{exm2} the integral w.r.t.~$\eta^\P$ reduces to a sum, which is evaluated by interpolation between the grid points.

It is easier to verify convergence conditions from \citet{barles1991,dang2014} for the HJB equation that includes $s$ as a state variable than the reduced one. Convergence for the reduced equation follows by using  homogeneity of $V$ in $s$,.

Let $\wstep,\pistep$ be the step-sizes in $w,\pi^1$-direction. In our computations we choose $\wstep=10$ and $\pistep=1/20$. For fixed time point $t_k$ we determine the control $\nu=\nu_{t_k}^\ast(s,w,\pi)$ by maximizing $H$, see \eqref{eq:optprob}. The discretized version of the HJB equation is
\begin{align}
V(t_{k+1},s,w,\pi^1)&=V(t_k,s,w,\pi^1)+(t_{k+1}-t_k)\Big[-\zeta_1(t_k,w,\pi^1) V (t_k,s,w,\pi^1)\\
&+ \zeta_2(t_k,w,\pi^1) V(t_k,s,w-\wstep,\pi^1)\\
&+\frac{1}{2}(\zeta_3(t_k,w,\pi^1) + |\zeta_3(t_k,w,\pi^1)|)V(t_k,s,w,\pi^1+\pistep)\\
&+\frac{1}{2}(\zeta_3(t_k,w,\pi^1) - |\zeta_3(t_k,w,\pi^1)|)V(t_k,s,w,\pi^1-\pistep)\\
&+\lambda_1(t_k,w,\pi^1) V\left(t_k,s(1-\theta),w,\frac{\pi^1 c^\text{down}_1}{\pi^1 c^\text{down}_1+(1-\pi^1) c^\text{down}_2}  \right)\\
&+ \lambda_2(t_k,w,\pi^1) V\left(t_k,s(1+\theta),w,\frac{\pi^1 c^\text{up}_1}{\pi^1 c^\text{up}_1+(1-\pi^1) c^\text{up}_2} \right) \Big]\\
&+(t_{k+1}-t_k) (\nu-\fconst \fexp \nu^{\fexp+1})s\,,\label{eq:scheme-s}
\end{align}
where
\begin{align*}
\zeta_2(t_k,w,\pi^1)&= \frac{\nu}{\wstep} \,,\\
\zeta_3(t_k,w,\pi^1)&= (\pi^1 q^{11}+(1-\pi^1)q^{21}) - \pi^1(1-\pi^1)\left((1+a\nu_{t_k}^\ast(w,\pi^1))(c^\text{down}_1- c^\text{down}_2)+ c^\text{up}_1 - c^\text{up}_2 \right) \,,\\
\lambda_1(t_k,w,\pi^1)&= (\pi^1 c^\text{down}_1+(1-\pi^1)c^\text{down}_2)(1+a\,\nu) \,,\\
\lambda_2(t_k,w,\pi^1)&= (\pi^1 c^\text{up}_1+(1-\pi^1)c^\text{up}_2) \,,\\
\zeta_1(t_k,w,\pi^1)&=\rho + \zeta_2(t_k,w,\pi^1) + |\zeta_3(t_k,w,\pi^1)| +\lambda_1(t_k,w,\pi^1)+\lambda_2(t_k,w,\pi^1) \,.
 \end{align*}
 On the active boundary we set $V=h$.
Due to the use of an upwinding scheme, on the non-active part of the boundary the solution to \eqref{eq:scheme-s} is determined endogenously, i.e.~by the values of $V$ in the interior of the state space. Note that this is in line with the formulation of the boundary conditions for the limiting equation in Definition \ref{def:viscosity-1}.

By construction this scheme is consistent.
In order to get stability we need to choose the time step
sufficiently small, namely
$t_{k+1}-t_{k} \le 1/\zeta_1(t_k,w,\pi^1)$.
Monotonicity requires positivity of the coefficients of the difference operators and suitable quadrature weights for the integral term (for the latter see \citet{dang2014}).
In our context, this holds since $\zeta_2, \lambda_1, \lambda_2$ are positive,
and the coefficients of $V(t_k,s,w,\pi^1+\pistep), V(t_k,s,w,\pi^1-\pistep)$ are positive by construction (the construction we used is an easy implementable way of decomposing the coefficient $\zeta_3$ into its positive and its negative part). Hence, $\zeta_1$ is also positive.
Since we also showed in Theorem \ref{thm:viscosity} that the comparison principle holds and that the value function is the unique viscosity solution to our HJB equation, we get convergence of the proposed scheme to the value function.
%


\begin{thebibliography}{34}
\providecommand{\natexlab}[1]{#1}
\providecommand{\url}[1]{\texttt{#1}}
\expandafter\ifx\csname urlstyle\endcsname\relax
  \providecommand{\doi}[1]{doi: #1}\else
  \providecommand{\doi}{doi: \begingroup \urlstyle{rm}\Url}\fi

\bibitem[Almgren and Chriss(2001)]{almgren2001optimal}
R.~Almgren and N.~Chriss.
\newblock Optimal execution of portfolio transactions.
\newblock \emph{Journal of Risk}, 3:\penalty0 5--40, 2001.

\bibitem[Almgren et~al.(2005)Almgren, Thum, Hauptmann, and Li]{almgren}
R.~Almgren, C.~Thum, E.~Hauptmann, and H.~Li.
\newblock Direct estimation of equity market impact.
\newblock \emph{Risk}, 18\penalty0 (5752):\penalty0 10, 2005.

\bibitem[Almudevar(2001)]{almudevar2001dynamic}
A.~Almudevar.
\newblock A dynamic programming algorithm for the optimal control of piecewise
  deterministic {Markov} processes.
\newblock \emph{SIAM Journal on Control and Optimization}, 40\penalty0
  (2):\penalty0 525--539, 2001.

\bibitem[Andersen(1996)]{bib:andersen-96}
T.~G. Andersen.
\newblock Return volatility and trading volume: {A}n information flow
  interpretation of stochastic volatility.
\newblock \emph{The Journal of Finance}, 51\penalty0 (1):\penalty0 169--204,
  1996.

\bibitem[Bain and Crisan(2009)]{bain2009filtering}
A.~Bain and D.~Crisan.
\newblock \emph{Fundamentals of Stochastic Filtering}, volume~3.
\newblock Springer, 2009.

\bibitem[Barles(1994)]{bib:barles-94}
G.~Barles.
\newblock \emph{Solutions de Viscosit{\'e} des {E}quations de Hamilton-Jacobi}.
\newblock Springer Verlag, 1994.

\bibitem[Barles and Souganidis(1991)]{barles1991}
G.~Barles and P.~E. Souganidis.
\newblock {Convergence of approximation schemes for fully nonlinear second
  order equations}.
\newblock \emph{Asymptotic Analysis}, 4:\penalty0 271--283, 1991.

\bibitem[B{\"a}uerle and Rieder(2009)]{bib:bauerle-rieder-09}
N.~B{\"a}uerle and U.~Rieder.
\newblock {MDP} algorithms for portfolio optimization problems in pure jump
  markets.
\newblock \emph{Finance and Stochastics}, 13\penalty0 (4):\penalty0 591--611,
  2009.

\bibitem[B{\"a}uerle and Rieder(2011)]{bauerle-rieder-book}
N.~B{\"a}uerle and U.~Rieder.
\newblock \emph{Markov Decision Processes with Applications to Finance}.
\newblock Springer Science \& Business Media, 2011.

\bibitem[Bertsimas and Lo(1998)]{bertsimas1998optimal}
D.~Bertsimas and A.~W. Lo.
\newblock Optimal control of execution costs.
\newblock \emph{Journal of Financial Markets}, 1\penalty0 (1):\penalty0 1--50,
  1998.

\bibitem[Br\'emaud(1981)]{bremaud1981point}
P.~Br\'emaud.
\newblock \emph{Point Processes and Queues: Martingale Dynamics}.
\newblock Springer Series in Statistics. Springer-Verlag, New York Heidelberg
  Berlin, 1981.

\bibitem[Cartea and Jaimungal(2013)]{bib:cartea-jaimungal-13}
A.~Cartea and S.~Jaimungal.
\newblock Modelling asset prices for algorithmic and high frequency trading.
\newblock \emph{Applied Mathematical Finance}, 20\penalty0 (6):\penalty0
  512--547, 2013.

\bibitem[Cartea et~al.(2015)Cartea, Jaimungal, and
  Penalva]{bib:cartea-jaimungal-penalva-15}
{\'A}.~Cartea, S.~Jaimungal, and J.~Penalva.
\newblock \emph{Algorithmic and high-frequency trading}.
\newblock Cambridge University Press, 2015.

\bibitem[Casgrain and Jaimungal(2019)]{bib:casgrain-jaimungal-19}
P.~Casgrain and S.~Jaimungal.
\newblock Trading algorithms with learning in latent alpha models.
\newblock \emph{Mathematical Finance}, 29:\penalty0 735--772, 2019.

\bibitem[Ceci and Colaneri(2012)]{ceci2012nonlinear}
C.~Ceci and K.~Colaneri.
\newblock Nonlinear filtering for jump diffusion observations.
\newblock \emph{Advances in Applied Probability}, 44\penalty0 (3):\penalty0
  678--701, 2012.

\bibitem[Ceci and Colaneri(2014)]{ceci2014zakai}
C.~Ceci and K.~Colaneri.
\newblock The {Z}akai equation of nonlinear filtering for jump-diffusion
  observation: existence and uniqueness.
\newblock \emph{Applied Mathematics and Optimization}, 69\penalty0
  (1):\penalty0 47--82, 2014.

\bibitem[Cont(2011)]{bib:cont-11}
R.~Cont.
\newblock Statistical modeling of high-frequency financial data.
\newblock \emph{IEEE Signal Processing Magazine}, 28\penalty0 (5):\penalty0
  16--25, 2011.

\bibitem[Costa and Dufour(2013)]{costa2013continuous}
O.~L. Costa and F.~Dufour.
\newblock \emph{Continuous Average Control of Piecewise Deterministic Markov P
  rocesses}.
\newblock Springer, 2013.

\bibitem[Cvitanic et~al.(2006)Cvitanic, Rozovskii, and
  Zaliapin]{bib:cvitanic-rozovski-zalyapin-05}
J.~Cvitanic, B.~Rozovskii, and I.~Zaliapin.
\newblock Numerical estimation of volatility values from discretely observed
  diffusion data.
\newblock \emph{Journal of Computational Finance}, 9\penalty0 (4):\penalty0 1,
  2006.

\bibitem[Damian et~al.(2018)Damian, Eksi, and Frey]{bib:damian-eksi-frey-17}
C.~Damian, Z.~Eksi, and R.~Frey.
\newblock {EM} algorithm for {M}arkov chains observed via {Gaussian} noise and
  point process infrmation: {T}heory and numerical experiments.
\newblock \emph{Statistics and Risk Modelling}, 35:\penalty0 51--72, 2018.

\bibitem[Dang and Forsyth(2014)]{dang2014}
D.~M. Dang and P.~A. Forsyth.
\newblock {Continuous Time Mean-Variance Optimal Portfolio Allocation Under
  Jump Diffusion: A Numerical Impulse Control Approach}.
\newblock \emph{Numerical Methods for Partial Differential Equations},
  30:\penalty0 664--698, 2014.

\bibitem[Davis(1993)]{davis1993markov}
M.~{H.}~{A.} Davis.
\newblock \emph{Markov Models \& Optimization}, volume~49.
\newblock CRC Press, 1993.

\bibitem[Davis and Farid(1999)]{bib:davis-farid-99}
Mark~{H}.{A}. Davis and M.~Farid.
\newblock Piecewise-deterministic processes and viscosity solutions.
\newblock In \emph{Stochastic analysis, control, optimization and
  applications}, pages 249--268. Springer, 1999.

\bibitem[Fleming and Soner(2006)]{bib:fleming-soner-06}
W.H. Fleming and H.M. Soner.
\newblock \emph{Controlled Markov Processes and Viscosity Solutions}.
\newblock Springer, New York, 2nd edition, 2006.

\bibitem[Frey and Schmidt(2012)]{frey2012pricing}
R.~Frey and T.~Schmidt.
\newblock Pricing and hedging of credit derivatives via the innovation approach
  to nonlinear filtering.
\newblock \emph{Finance and Stochastics}, 16\penalty0 (1):\penalty0 105--133,
  2012.

\bibitem[Gatheral and Schied(2011)]{gatheral2011optimal}
J.~Gatheral and A.~Schied.
\newblock Optimal trade execution under geometric {Brownian} motion in the
  {Almgren and Chriss} framework.
\newblock \emph{International Journal of Theoretical and Applied Finance},
  14\penalty0 (03):\penalty0 353--368, 2011.

\bibitem[Gatheral and Schied(2013)]{gatheral2013dynamical}
J.~Gatheral and A.~Schied.
\newblock Dynamical models of market impact and algorithms for order execution.
\newblock In J.P. Jean-Pierre~Fouque and J.A. Langsam, editors, \emph{Handbook
  on Systemic Risk}, pages 579--599. 2013.

\bibitem[Guo and Zervos(2015)]{guo2015optimal}
X.~Guo and M.~Zervos.
\newblock Optimal execution with multiplicative price impact.
\newblock \emph{SIAM Journal on Financial Mathematics}, 6\penalty0
  (1):\penalty0 281--306, 2015.

\bibitem[He and Mamaysky(2005)]{he2005dynamic}
H.~He and H.~Mamaysky.
\newblock Dynamic trading policies with price impact.
\newblock \emph{Journal of Economic Dynamics and Control}, 29\penalty0
  (5):\penalty0 891--930, 2005.

\bibitem[Jacod and Shiryaev(2003)]{jacod2003limit}
J.~Jacod and A.N. Shiryaev.
\newblock \emph{Limit Theorems for Stochastic Processes}.
\newblock Springer, 2nd edition, 2003.

\bibitem[Lehalle et~al.(2018)Lehalle, Mounjid, and
  Rosenbaum]{lehalle2018optimal}
C.A. Lehalle, O.~Mounjid, and M.~Rosenbaum.
\newblock Optimal liquidity-based trading tactics.
\newblock \emph{arXiv preprint arXiv:1803.05690}, 2018.

\bibitem[Protter and Shimbo(2008)]{protter2008arbitrage}
P.~Protter and K.~Shimbo.
\newblock No arbitrage and general semimartingales.
\newblock \emph{Markov Processes and Related Topics: A Festschrift for Thomas
  G. Kurtz}, 4:\penalty0 267--283, 2008.

\bibitem[Schied(2013)]{schied16}
A.~Schied.
\newblock Robust strategies for optimal order execution in the
  {Almgren--Chriss} framework.
\newblock \emph{Applied Mathematical Finance}, 20\penalty0 (3):\penalty0
  264--286, 2013.

\bibitem[Schied and Sch{\"o}neborn(2009)]{schied2009risk}
A.~Schied and T.~Sch{\"o}neborn.
\newblock Risk aversion and the dynamics of optimal liquidation strategies in
  illiquid markets.
\newblock \emph{Finance and Stochastics}, 13\penalty0 (2):\penalty0 181--204,
  2009.

\end{thebibliography}
\end{document}